\documentclass[preprint,dvips]{imsart}

\usepackage[round,comma,authoryear]{natbib}
\usepackage{placeins}
\usepackage{array}
\usepackage[latin1]{inputenc} 
\usepackage[british]{babel}
\usepackage{latexsym}
\usepackage[dvips]{graphics,graphicx,psfrag}
\usepackage{amssymb,amsmath,amsfonts,amsthm}
\usepackage{verbatim}
\usepackage{pstricks,pst-node,psfrag}
\usepackage[breaklinks]{hyperref}
\usepackage{epsfig}

\graphicspath{{figs/}{./}}

\startlocaldefs

\DeclareMathOperator{\R}{\mathbb{R}}

\DeclareMathOperator{\N}{\mathbb{N}}

\DeclareMathOperator{\noise}{\mathcal{W}}

\DeclareMathOperator{\om}{\omega}
\DeclareMathOperator{\pd}{\partial}

\DeclareMathOperator{\diag}{diag}

\theoremstyle{plain} 
\newtheorem{thm}{Theorem}[section]

\newtheorem{prop}[thm]{Proposition}

\newtheorem{alg}[thm]{Algorithm}
\theoremstyle{definition} 

\theoremstyle{remark} 
\newtheorem{rem}{Remark}

\newcommand{\proper}{\mathsf}
\newcommand{\pP}{\proper{P}}
\newcommand{\pE}{\proper{E}}

\newcommand{\pC}{\proper{C}}
\newcommand{\pN}{\proper{N}}
\newcommand{\mv}[1]{{\boldsymbol{\mathrm{#1}}}}
\newcommand{\trsp}{\ensuremath{\top}}
\newcommand{\md}{\ensuremath{\,\mathrm{d}}}
\newcommand{\scal}[2]{\left\langle {#1},\,{#2} \right\rangle}
\endlocaldefs
\begin{document}

\begin{frontmatter}

\title{Spatial Mat\'{e}rn fields driven by non-Gaussian noise}
\runtitle{Spatial Mat\'{e}rn fields driven by non-Gaussian noise}

\author{\fnms{David} \snm{Bolin}\corref{}\ead[label=e1]{bolin@maths.lth.se}}
\address{Mathematical Statistics\\
Centre for Mathematical Sciences\\
Lund University\\
Sweden\\\printead{e1}}
\affiliation{Lund University}
\runauthor{David Bolin}

\begin{abstract}
The article studies non-Gaussian extensions of a recently discovered link between certain Gaussian random fields, expressed as solutions to stochastic partial differential equations (SPDEs), and Gaussian Markov random fields. The focus is on non-Gaussian random fields with Mat\'ern covariance functions, and in particular we show how the SPDE formulation of a Laplace moving average model can be used to obtain an efficient simulation method as well as an accurate parameter estimation technique for the model. This should be seen as a demonstration of how these techniques can be used, and generalizations to more general SPDEs are readily available. 
\end{abstract}



\end{frontmatter}

\section{Introduction}
Recently, \cite{lindgren10} derived a link between certain Gaussian fields, that can be represented as solutions to stochastic partial differential equations (SPDEs), and Gaussian Markov random fields (GMRFs). The main idea is to approximate these Gaussian fields using basis expansions $\sum_i w_i\varphi_i(s)$ where the stochastic weights $\{w_i\}$ are calculated using the stochastic weak formulation of the corresponding SPDE. For certain choices of the basis functions $\{\varphi_i\}$, especially compactly supported functions, the weights  form GMRFs. Because of the Markov property of the weights, fast numerical techniques for sparse matrices can be used when estimating parameters and doing spatial prediction in these models. This greatly improves the applicability to problems involving large data sets, where traditional methods in statistics fail due to computational issues. However, the advantages of representing Gaussian fields as solutions to SPDEs are not only computational.  Using the SPDE representation, non-stationary extensions are easily obtained by allowing spatially varying parameters in the SPDE \citep{lindgren10}, and the model class can be generalized to include more general covariance structures by generalizing the class of generating SPDEs \citep{bolin09b}. These are indeed useful features from an applied point of view as many applications require complicated non-stationary models to accurately capture the covariance structure of the data. 

So far these methods have only been used in Gaussian settings, and it has not been clear whether they are applicable when the Gaussianity assumption cannot be justified. Therefore, this work will focus on extending the SPDE methods beyond Gaussianity. A new type of non-Gaussian models that has proved to be useful in practical applications is the Laplace moving average models \citep{aberg08,aberg08b}. These are processes obtained by convolving some deterministic kernel function with stochastic Laplace noise. 
The models share many good properties with the Gaussian models while allowing for heavier tails and asymmetry in the data, making them interesting alternatives in practical applications \citep[see e.g.][]{bogsjo12}. One of the motivating examples in \citet{aberg08b} is a Laplace moving average model with Mat\'{e}rn covariances. This model can be seen as the solution to the same SPDE that generates Gaussian Mat\'{e}rn field but where the Gaussian white noise forcing is replaced with Laplace noise. It has previously been shown that the SPDE model formulation of Gaussian Mat\'{e}rn fields has many computational advantages compared with the process convolution formulation \citep{bolin09,simpson10}. We demonstrate here that for the Laplace moving average models, the SPDE formulation can also be used to derive a new likelihood-based parameter estimation technique as well as an efficient simulation procedure.

The structure of the paper is as follows. Section \ref{paperD:sec:matern} contains an introduction to the Mat\'{e}rn covariance family and the SPDE formulation in the Gaussian case. In Section \ref{paperD:sec:laplace}, stochastic Laplace fields are introduced, and some properties of the Laplace-driven SPDE model are derived. Subsequently, in Section \ref{paperD:sec:hilbert}, the Markov approximation technique by \cite{lindgren10} is extended to the Laplace model, and its sampling is discussed in Section \ref{paperD:sec:sampel}. A parameter estimation technique based on the EM algorithm is derived in Section \ref{paperD:sec:estimation}, and Section \ref{paperD:sec:study} contains a simulation study showing that it gives reliable parameter estimates. Finally, Section \ref{paperD:sec:discussion} contains a summary and discussion of future work and possible extensions.

\section{Gaussian Mat\'{e}rn fields}\label{paperD:sec:matern}
The Mat\'{e}rn covariance family \citep{matern60} is often used when modeling spatial data. There are a few different parameterizations of the Mat\'{e}rn covariance function in the literature, and the one most suitable in our context is
\begin{equation}\label{paperD:eq:matern}
C(\mv{h}) = \frac{2^{1-\nu}\phi^2}{(4\pi)^{\frac{d}{2}}\Gamma(\nu + \frac{d}{2})\kappa^{2\nu}}(\kappa\|\mv{h}\|)^{\nu}K_{\nu}(\kappa\|\mv{h}\|), \quad \mv{h} \in \R^d,
\end{equation}
where $d$ is the dimension of the domain, $\nu$ is a shape parameter, $\kappa^2$ a scale parameter, $\phi^2$ a variance parameter, and $K_{\nu}$ is a modified Bessel function of the second kind of order $\nu>0$. The associated spectrum is
\begin{equation}\label{paperD:matern_spec}
S(\mv{k}) = \frac{\phi^2}{(2\pi)^d}\frac{1}{(\kappa^2+\mv{k}^{\trsp}\mv{k})^{\nu + \frac{d}{2}}}.
\end{equation}
As the properties of Gaussian fields are given by their first two moments, the standard way of specifying Gaussian Mat\'{e}rn fields is to chose the mean value, $\mu(\mv{s})$ possibly spatially varying, and then let the covariance function be of the form \eqref{paperD:eq:matern}. An alternative way of specifying a Gaussian field on $\R^d$ is to view it as a process convolution
\begin{equation}\label{paperD:eq:gaussianconv}
X(\mv{s}) = \int_{\R^d} k(\mv{s},\mv{u}) \mathcal{B}(\md \mv{u}),
\end{equation}
where $k$ is some deterministic kernel function and $\mathcal{B}$ is a Brownian sheet \citep{higdon01}. One of the advantages with this construction is that non-stationary extensions are easily constructed by allowing the convolution kernel to be dependent on the location $\mv{s}$. If, however, the process is stationary, the kernel $k$ depends only on $\mv{s}-\mv{u}$ and the covariance function for $X$ is
\begin{equation*}
C(\mv{h}) = \int_{\R^d}k(\mv{u}-\mv{h})k(\mv{u})\md\mv{u}.
\end{equation*}
Thus, the covariance function $C$, the spectrum $S$, and the kernel $k$ are related through 
\begin{equation*}
(2\pi)^{d}|\mathcal{F}(k)|^2 = \mathcal{F}(C) = S,
\end{equation*}
where $\mathcal{F}(\cdot)$ denotes the Fourier transform. Since the spectral density for a Mat\'{e}rn field in dimension $d$ with parameters $\nu$, $\phi^2$, and $\kappa$ is given by \eqref{paperD:matern_spec}, one finds that the corresponding symmetric non-negative kernel is a Mat\'{e}rn covariance function with parameters $\nu_k = \frac{\nu}{2}-\frac{d}{4}$, $\phi_k = \sqrt{\phi}$, and $\kappa_k = \kappa$.

In yet another setting, Gaussian Mat\'{e}rn fields can be viewed as the solution to the SPDE
\begin{equation}\label{paperD:spde}
(\kappa^2-\Delta)^{\frac{\alpha}{2}}X(\mv{s}) = \phi \noise(\mv{s}),
\end{equation}
where $\noise(\mv{s})$ is Gaussian white noise, $\Delta = \sum_{i=1}^d \frac{\pd^2}{\pd\mv{s}_i^2}$ is the Laplace operator, and $\alpha = \nu + d/2$ \citep{whittle63}. As discussed in \cite{lindgren10}, there is an implicit assumption of appropriate boundary conditions needed if one wants the solutions to be stationary Mat\'{e}rn fields. 

The connection between \eqref{paperD:eq:gaussianconv} and \eqref{paperD:spde} is through the Green's function of the differential operator in \eqref{paperD:spde}
\begin{equation}\label{paperD:eq:green}
G_{\alpha}(\mv{s},\mv{t}) = \frac{2^{1-\frac{\alpha - d}{2}}}{(4\pi)^{\frac{d}{2}}\Gamma(\frac{\alpha}{2})\kappa^{\alpha - d}}(\kappa\|\mv{s}-\mv{t}\|)^{\frac{\alpha -d}{2}}K_{\frac{\alpha-d}{2}}(\kappa\|\mv{s}-\mv{t}\|),
\end{equation}
that serves as a kernel in \eqref{paperD:eq:gaussianconv}. It is straightforward to show that $G_{\alpha}\in L_p(\R^d)$ if and only if $\alpha > \frac{(p-1)d}{p}$ (see for example \cite{samko92} p.~538), and in particular $\alpha>d/2$ guarantees that $G_{\alpha}\in L_2(\R^d)$. 

A non-Gaussian model with Mat\'{e}rn covariances could be constructed either using the process convolution formulation \eqref{paperD:eq:gaussianconv} where the Brownian sheet is replaced by some non-Gaussian process, or through the SPDE formulation \eqref{paperD:spde} with non-Gaussian noise. Such non-Gaussian extensions are discussed next.

\section{Non-Gaussian SPDE-based models}\label{paperD:sec:laplace}
A simple way of moving beyond Gaussianity in the SPDE model \eqref{paperD:spde} is to allow for a stochastic variance parameter $\phi$. By choosing $\phi$ as an inverse-gamma distributed random variable, the resulting field has t-distributed marginal distributions and is therefore sometimes referred to as a t-distributed random field \citep{roislien06}. In a Bayesian setting, this extension can be interpreted simply as choosing a certain prior distribution for the variance, and one can of course come up with many other non-Gaussian models by changing this distribution. However, models constructed in this way are non-Gaussian only in a very limited sense. Namely, every realization of them behaves exactly as a Gaussian field with a globally re-scaled variance, and because of this, they are all non-ergodic as the parameters in the prior distribution cannot be estimated from a single realization of the field. One would prefer a non-Gaussian model where the actual sample paths behave differently from a stationary Gaussian field, and one way of achieving this is to let the variance parameter be spatially and stochastically varying. Both \cite{lindgren10} and \cite{bolin09b} explores this option by expressing $\log\phi(\mv{s})$ as a regression on a few known basis functions where the stochastic weights are estimated from data. This was interpreted as a non-stationary Gaussian model, but could also be viewed as a, somewhat limited, non-Gaussian model with a slowly spatially varying variance parameter $\phi(\mv{s})$. To obtain a model which is intrinsically non-Gaussian also within realizations, one can draw $\phi(\mv{s})$ at random independently for each $\mv{s}$. The right-hand side of \eqref{paperD:spde} is then a product of two independent noise fields. 
The following non-Gaussian models essentially can be interpreted as a formal realization of this idea.

One interesting type of distributions, obtained by taking a random variance and mean in an otherwise Gaussian random variable, are the generalized asymmetric Laplace distributions \citep{aberg08}. The Laplace distribution is defined through the characteristic function with parameters $\mu,\gamma \in \R$ and $\sigma,\tau>0$
\begin{equation*}
\varphi(u) = e^{i\gamma u}\left(1-i\mu u + \frac{\sigma^2}{2}u^2\right)^{-\tau}.
\end{equation*}
The distribution is symmetric if $\mu=0$ and asymmetric otherwise. The shape of the distribution is governed by $\tau$ and the scale by $\sigma$. The distribution is infinitely divisible, and a useful characterization is that if $Z$ is a standard normal variable and $\Gamma$ is an independent gamma variable with shape $\tau$, then $\gamma + \mu\Gamma + \sigma\sqrt{\Gamma}Z$ has an asymmetric Laplace distribution. 

Stochastic Laplace noise can now be obtained from an independently scattered random measure $\Lambda$, defined for a Borel set $B$ in $\R^d$ by the characteristic function
\begin{equation*}
\varphi_{\Lambda(B)}(u) = e^{i\gamma m(B) u}\left(1-i\mu u + \frac{\sigma^2}{2}u^2\right)^{- m(B)},
\end{equation*}
where the measure $m$ is referred to as the control measure of $\Lambda$. This does not define Laplace noise in a direct manner, but similarly to how Gaussian white noise can be seen as a differentiated Brownian sheet \citep{walsh84}, Laplace noise can be viewed in the sense of distributions (generalized functions) as a differentiated Laplace field. The most transparent characterization is through the following series representation of the Laplace field $\Lambda(\mv{s})$ on a compact set ${D\in \R^d}$:
\begin{equation}\label{paperD:eq:lfieldseries}
\Lambda(\mv{s}) = \gamma \mv{s} + \sum_{k=1}^{\infty}\left(\Gamma_k + G_k\sqrt{\Gamma_k}\right)\mv{1}(\mv{s}\geq\mv{s}_k), \,\,\,\mv{s}\in D,
\end{equation}
where $G_k$ are iid $\pN(0,1)$ random variables, $\mv{s}_k$ are iid uniform random variables on $D$, and 
\begin{equation*}
\mv{1}(\mv{s}\geq \mv{s}_k) = \begin{cases}1 & \mbox{if $s_i\geq s_{k,i}$ for all $i \leq d$},\\
0 & \mbox{otherwise.}\end{cases}
\end{equation*}
The random variables $\Gamma_k$ can be written as $\Gamma_k = e^{-\nu\gamma_k}W_k$ where $W_k$ are iid standard exponential variables and $\gamma_k$ are the arrival times of a Poisson process with intensity 1.
Thus, Laplace noise can be expressed as a distribution (generalized function) 
\begin{equation}\label{paperD:eq:lnoiseseries}
\dot{\Lambda} = \gamma + \sum_{k=1}^{\infty}\left(\Gamma_k + G_k\sqrt{\Gamma_k}\right)\delta_{\mv{s}_k},
\end{equation}
where $\delta_{\mv{s}_k}$ is the Dirac delta distribution centered at $\mv{s}_k$. 

The model of interest is the solution $X$ to the Laplace-driven SPDE
\begin{equation}\label{paperD:eq:lspde}
(\kappa^2 - \Delta)^{\frac{\alpha}{2}}X = \dot{\Lambda},
\end{equation}
where both $X$ and $\dot{\Lambda}$ are viewed as random variables valued in the space of tempered distributions. To clarify in what way the solution to this equation exists, we look at a general SPDE
\begin{equation}\label{paperD:eq:lspdep}
(\kappa^2 - \Delta)^{\frac{\alpha}{2}}X = \dot{M},
\end{equation}
where $M$ is an arbitrary independently scattered $L_2$-valued random measure with $\pE(|M(\md \mv{x})|^2) = C\md \mv{x}$ for some constant $C<\infty$. Examples of such measures are the Laplace measures of interest here but also standard Brownian sheets. As usual for fractional Laplacian operators \citep{samko92}, ${\mathcal{T} = (\kappa^2 - \Delta)^{\frac{\alpha}{2}}}$ is defined using the Fourier transform through $\mathcal{F}(\mathcal{T}f) = \mathcal{P}\hat{f}$, where $\hat{f}$ is the Fourier transform of the function $f$, ${(\mathcal{P}\hat{f})(\mv{k}) = (\kappa^2+ \mv{k}^{\trsp}\mv{k})^{\frac{\alpha}{2}}\hat{f}(\mv{k})}$, and the operator $\mathcal{T}$ is well-defined for example for all $f \in L_p(\R^d)$ for ${1\leq p\leq \infty}$. The definition applies also when $f$ is a distribution or, more specifically, a tempered distribution. Thus, \eqref{paperD:eq:lspdep} is viewed as an equation for two random (tempered) distributions so the equation has to be interpreted in the weak sense
\begin{equation}\label{paperD:eq:weak1}
\mathcal{T}X(\varphi) = \dot{M}(\varphi),
\end{equation}
where $\varphi$ is in some appropriate space of test functions. Now, the action of the self-adjoint operator $\mathcal{T}$ can be moved to the test function on the left-hand side and \eqref{paperD:eq:weak1} can be rewritten in a more explicit fashion as 
\begin{equation}\label{paperD:eq:p1}
X(\mathcal{T}\varphi,\om) = \int \varphi(\mv{s}) M(\md\mv{s},\om).
\end{equation}
Here, we have included the second argument $\om\in\Omega$ to highlight that the sought functional $X$ is random, and the equation should hold for $\omega$ in a certain full probability set $\Omega_0\in\Omega$ and universally for each $\varphi$.

To describe the solutions of \eqref{paperD:eq:lspdep}, we need the Sobolev spaces $H_n$ of fractional order $n$. These are usually defined using the Fourier transform in the following way. Let $E$ be the Schwartz space of rapidly decreasing functions on $\R^d$, for $u\in E'$ (the dual of $E$, also referred to as the space of tempered distributions), define the Fourier transform of $u$ as $\hat{u}(\varphi) = u(\hat{\varphi})$, where $\hat{\varphi}$ is the usual Fourier transform on $\R^d$ of $\varphi\in E$. Define a norm on $E$ by
\begin{equation*}
\|u\|_n = \int_{\R^d} (1+|\mv{k}|^2)^n |\hat{u}(\mv{k})|^2 \md\mv{k}
\end{equation*}
and let $H_n$ be the completion of $E$ in this norm. By Plancherel's theorem, one has that $H_0 = L_2(\R^d)$ and one can show that for the special case $n\in\N$, $H_n$ is identical to the classical Sobolev space of $L_2$ functions with all partial derivatives of order $n$ or less in $L_2$. The space $H_{-n}$ is the dual space of $H_{n}$ and does in general contain distributions.

Let us note that the right hand side of \eqref{paperD:eq:p1} in principle may not be defined on a full probability set uniformly for all $\varphi$. However, one can regularize $M$ so that $\varphi \rightarrow M(\varphi)$ is in fact a random distribution. Indeed, since 
\begin{equation*}
\pE(|M(\varphi)|^2) = C\int\varphi(\mv{s})^2\md\mv{s} = C\|\varphi\|_0^2,
\end{equation*}
the random linear functional $\varphi \rightarrow M(\varphi)$ is continuous in probability on $H_n$ for any $n\geq0$, and by Theorem 4.1 in \cite{walsh84} there exists a version of $M$ which is almost surely in $H_{-n}$ for $n>d/2$. From now on we always assume that we deal with such a version.

Following \cite{walsh84}, we say that $X(\cdot,\om)$ is an $H_n$-solution of \eqref{paperD:eq:lspdep} if for a.e. $\om$, $X(\cdot,\om)$ is an element of $H_{-n}$ and \eqref{paperD:eq:p1} holds for every $\varphi\in H_n$. In other words, we aim at finding a random functional $X$ that almost surely is a distribution and satisfies \eqref{paperD:eq:lspdep} as a continuous functional on $H_n$. The proof of the following proposition is similar to the proof of Proposition 9.1 in \cite{walsh84} where the existence of the solution to the stochastic Poisson equation on a bounded domain in $\R^d$ was demonstrated.

\begin{prop}\label{paperD:prop1}
Assume that $M$ is an independently scattered $L_2$-valued random measure with $\pE(|M(\md \mv{x})|^2) = C\md \mv{x}$. Then for $\kappa > 0$, $\alpha > 0$, there exists a random functional $X: H_n \times \Omega \rightarrow \R$ such that for a certain set $\Omega_0$, $P(\Omega_0)=1$ and for all $\omega \in \Omega_0$ and all $\varphi\in H_n$ 
\begin{equation}\label{paperD:eq:solution}
X(\varphi,\om) = \int G^{\alpha}\varphi(\mv{x})M(\md \mv{x},\om),
\end{equation}
where $G^\alpha\varphi(\mv{x})=\int G_\alpha(\mv{s},\mv{x}) \varphi(\mv{s})\md\mv{s}$ and $G_{\alpha}$ is given by \eqref{paperD:eq:green}. This is the unique $H_n$-solution to \eqref{paperD:eq:lspdep} if $n>d/2$, and moreover we have  $X\in H_m$ almost surely for $m<\alpha-d/2$.
\end{prop}
\begin{proof}
From the standard theory of fractional differential equations, one has that $G^{\alpha}$ maps $H_n$ isomorphically onto $H_{n+\alpha}$ \cite[see e.g.][p.547]{samko92}. Let $X$ be any $H_n$-solution to \eqref{paperD:eq:lspdep} and let $\psi = G^{\alpha}\varphi$. Applying \eqref{paperD:eq:p1} to $\psi$ and using that $\mathcal{T}G^{\alpha}\varphi = \varphi$ one gets that 
\begin{equation*}
X(\varphi) = X(\mathcal{T}G^{\alpha}\varphi) = X(\mathcal{T}\psi) = \int \psi(\mv{y}) M(\md \mv{y}) = \int G^{\alpha}\varphi(\mv{y}) M(\md \mv{y}).
\end{equation*}
Thus this solution also satisfies \eqref{paperD:eq:solution} and the solution is unique if it exists. 

To prove existence, let $X$ be defined by \eqref{paperD:eq:solution} and take $\varphi\in L_2(\R^d)$. Then
\begin{align*}
\pE(|X(\varphi)|^2) &= \pE\left[\left(\int G^{\alpha}\varphi(\mv{y})M(\md \mv{y})\right)^2\right] \\ &= \int G^{\alpha}\varphi(\mv{x})G^{\alpha}\varphi(\mv{y})\pE\left[M(\md \mv{x})M(\md \mv{y})\right]\\
&= C\int \left(G^{\alpha}\varphi(\mv{x})\right)^2\md \mv{x} = C\|G^{\alpha}\varphi\|_0^2 \leq  C_2 \|\varphi\|_{-\alpha}^2,
\end{align*}
where the last inequality follows from that $G^{\alpha}$ maps $H_{n-\alpha} \rightarrow H_{n}$ boundedly for $\alpha>0$. Thus, it follows that $X$ is a random linear functional that is continuous in probability on $H_{-\alpha}$. The embedding maps $H_{n_1} \rightarrow H_{n_2}$ are of Hilbert-Schmidt type if $n_1>n_2+d/2$ \citep[see e.g.~Example 1a in][]{walsh84}, and using this with $n_2=-\alpha$ together with Theorem 4.1 in \cite{walsh84} one gets that there exists a version of $X$ which is almost surely in $H_{-n}$ if $n>d/2>d/2-\alpha$. From now on, $X$ is such a version and we note that $X\in H_m$ almost surely for $m<\alpha-d/2$.

What is left to show now is that $X$ with probability one satisfies \eqref{paperD:eq:p1} for each $\varphi\in H_n$ for $n>d/2$. To that end, first note that if $\varphi\in H_n$, then by the definitions of $\mathcal{T}$ and $G^{\alpha}$ one has
\begin{align*}
G^{\alpha}\mathcal{T}\varphi(\mv{s}) &= \int G_{\alpha}(\mv{y},\mv{s})\mathcal{T}\varphi(\mv{y})\md\mv{y} \\ &= \mathcal{F}^{-1}\left((\kappa^2+ \mv{k}^{\trsp}\mv{k})^{-\frac{\alpha}{2}}(\kappa^2+ \mv{k}^{\trsp}\mv{k})^{\frac{\alpha}{2}}\hat{\varphi}(\mv{k})\right)(\mv{s}) \\ &= \varphi(\mv{s}).
\end{align*}
Let $n>d/2$ and fix $\varphi\in H_n$. If $M(\varphi,\om)$ denotes the functional $\int \varphi(\mv{s}) M(\md\mv{s},\omega)$, one has by the definition of $X$ and by the equation above that 
\begin{equation*}
\int |X(\mathcal{T}\varphi,\om)- M(\varphi,\om)|^2 \md\pP(\om)=0.
\end{equation*}
Hence, there is a set $\Omega_{\varphi}\subset\Omega$ with $\pP(\Omega_{\varphi})=1$ such that for each $\omega \in \Omega_{\varphi}$ one has $X(\mathcal{T}\varphi,\omega) =  M(\varphi,\omega)$. Now, $H_n$ is separable, so we can chose a countable base $B=\{b_i\}_{i=1}^{\infty}$ in $H_n$ and define $\bar{\Omega}_0 = \cap_{i=1}^{\infty}\Omega_{b_i}$. Then equality holds for each $f\in B$ and for each $\om\in\bar{\Omega}_0$ and $\pP(\bar{\Omega}_0)=1$ by the countability of $B$. 

The map $\varphi \rightarrow \mathcal{T}\varphi \rightarrow X(\mathcal{T}\varphi)$ of $H_n \rightarrow H_{n-\alpha} \rightarrow \R$ is continuous since $X$ is continuous on $H_n$ for $n>d/2-\alpha$ and $\mathcal{T}$ is a continuous map from $H_n$ to $H_{n-\alpha}$. Thus, both $X(\mathcal{T}\cdot,\om)$ and $M(\cdot,\om)$ are continuous functionals on $H_n$ for $\om$ in some full probability set $\tilde{\Omega}_0$ and equality therefore holds in \eqref{paperD:eq:p1} for each $\varphi\in H_n$ for $\omega\in\Omega_0 = \bar{\Omega}_0\cap\tilde{\Omega}_0$ since $B$ is linearly dense in $H_n$.
\end{proof}

\begin{rem}
By similar arguments one can show that the solution $X$ defined in Proposition \ref{paperD:prop1} also is a solution to the SPDE \eqref{paperD:eq:lspdep} in the sense that with probability one $X\in E'$ and \eqref{paperD:eq:p1} holds for every $\varphi\in E$. This is, however, a weaker statement since $E = \cap_n H_n$ and $E' = \cup_n H_n$.
\end{rem}

\begin{rem}
The solution $X$ defined in Proposition \ref{paperD:prop1} is in general a random linear functional. However, it can be identified with a random function if $\alpha>d/2$ since $X\in H_m$ almost surely for $m<\alpha-d/2$. Using the relation between $\alpha$ and the parameter $\nu$ in the Mat\'ern covariance function, $\alpha=\nu+d/2$, we see that $X\in H_m$ almost surely for $m<\nu$. Thus, $\nu$ acts as a smoothness parameter for the solution since the sample paths almost surely will be differentiable if $\nu>1$, two times differentiable if $\nu>2$ etc. 
\end{rem}

\begin{rem}
The previous remark can be strengthened using the Sobolev embedding theorem which shows that $H_n$ can be embedded in the Hölder space $C_k^r(\R^d)$ where $n - (r+k) =d/2$ and $r\in(0,1)$ \citep[see e.g.][]{adams75}. The space $C_k^r(\R^d)$ consists of functions such that all partial derivatives up to order $k$ are continuous and such that the $k$th partial derivatives are Hölder continuous with exponent $r$. Thus, if $\nu>d/2$, we almost surely have $X\in C_k^r(\R^d)$ (after possibly redefining it on a set of measure zero) where $k$ is the integer part of $\nu-d/2$ and $r = \nu-d/2 - k$.
\end{rem}

We now go back to the special case of Laplace noise and since the main interest here is ordinary random fields with Mat\'ern covariance functions, we from now on assume that $\alpha>d/2$ in \eqref{paperD:eq:lspde}. 
One sometimes uses $m(A) = l(A)\tau$, where $l$ is the Lebesgue measure and $\tau$ some constant, as a control measure for $\Lambda$. By the definition of the differential operator $\mathcal{T}$, it is then easy to see that the spectrum for the solution $X$ is
\begin{equation*}
R_X(\mv{k}) = \frac{\tau(\sigma^2 + \mu^2)}{(2\pi)^d}\frac{1}{(\kappa^2 + \mv{k}^{\trsp}\mv{k})^{\alpha}}.
\end{equation*}
Thus, the covariance function for $X$ is a Mat\'{e}rn covariance of the form \eqref{paperD:eq:matern} with $\phi^2 = \tau(\sigma^2 + \mu^2)$. Since $X$ is Laplace noise convolved with a Green function, which also has the form of a Mat\'{e}rn covariance function, the model is equivalent to the Laplace moving average models in \cite{aberg08,aberg08b}. Thus, using Theorem 1 in \cite{aberg08b}, the marginal distribution for $X(\mv{s})$ is given by the characteristic function 
\begin{equation}\label{paperD:eq:charf}
\phi_{X}(u) = \exp\left(\tau\int i\gamma G_{\alpha}(\mv{s},\mv{t})u-\log\left(1-i\mu uG_{\alpha}(\mv{s},\mv{t}) + \frac{\sigma^2u^2}{2}G_{\alpha}^2(\mv{s},\mv{t})\right)\md\mv{t}\right).
\end{equation}
A few examples of the marginal distributions for symmetric and asymmetric cases are shown in Figure \ref{paperD:fig:pdf}.
\begin{figure}[t]
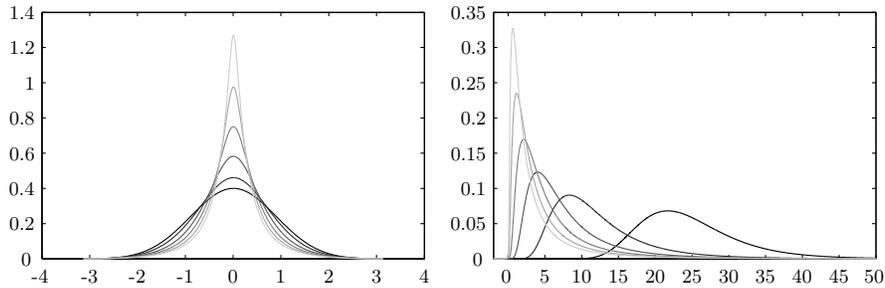

\begin{center}
\small
\resizebox{6cm}{!}{\input{figs/pdfsym.tex}}%
\resizebox{6cm}{!}{\input{figs/pdfasym.tex}}%
\end{center}
\vspace{-0.5cm}
\caption{Marginal distributions of the solution $X(\mv{s})$ to \eqref{paperD:eq:lspde} in the case of symmetric (left panel) and asymmetric (right panel) Laplace noise.}
\label{paperD:fig:pdf}
\end{figure}

\section{Hilbert space approximations}\label{paperD:sec:hilbert}
To obtain a computationally efficient representation of a Mat\'{e}rn field, the Hilbert space approximation technique by \cite{lindgren10} can be used. The starting point is to consider the stochastic weak formulation \eqref{paperD:eq:weak1} of the SPDE. A finite element approximation of the solution $X$ is then obtained by representing it as a finite basis expansion $\tilde{X}=\sum_{i=1}^n w_i \varphi_i(\mv{s})$, where the stochastic weights are calculated by requiring \eqref{paperD:eq:weak1} to hold for only a specific set of test functions $\{\psi_i,i=1,\ldots,n\}$ and $\{\varphi_i\}$ is a set of predetermined basis functions. To simplify the presentation, we first look at the case $\alpha/2 \in \N$ and then turn to the case of a general $\alpha > d/2$.  

\subsection{The case $\alpha/2 \in \N$}
To construct the approximation for $\alpha=2,4,\ldots$, we first look at the fundamental case $\alpha=2$. \cite{lindgren10} then use $\psi_i=\varphi_i$, and one then has
\begin{equation*}
   (\kappa^2-\Delta)\tilde{X}(\varphi_i) = \sum_{j=1}^n w_j\scal{\varphi_i}{(\kappa^2-\Delta)\varphi_j},
\end{equation*}
where $\scal{f}{g} = \int f(s)g(s)\md\mv{s}$. By introducing the vector $\mv{w} = (w_1, \ldots, w_n)^{\trsp}$ and a matrix $\mv{K}$ with elements ${\mv{K}_{ij} =
\scal{\varphi_i}{(\kappa^2-\Delta)\varphi_j}}$, the left hand side of
\eqref{paperD:eq:weak1} can be written as $\mv{K}\mv{w}$. Under mild conditions on the basis functions, one has
\begin{align*}
  \scal{\varphi_i}{(\kappa^2-\Delta)\varphi_j} &= \kappa^2\scal{\varphi_i}{\varphi_j}-\scal{\varphi_i}{\Delta\varphi_j}\\
  &=\kappa^2\scal{\varphi_i}{\varphi_j}+\scal{\nabla\varphi_i}{\nabla\varphi_j}.
\end{align*}
Hence, the matrix $\mv{K}$ can be written as the sum $\mv{K} = \kappa^2\mv{C} +
\mv{G}$ where $\mv{C}$ and $\mv{G}$ are matrices with elements $\mv{C}_{ij} = \scal{\varphi_i}{\varphi_j}$ and $\mv{G}_{ij} = \scal{\nabla\varphi_i}{\nabla\varphi_j}$ respectively.

\subsubsection{Gaussian noise}
In the Gaussian case, when $\dot{M}$ is Gaussian white noise, the right hand side of \eqref{paperD:eq:weak1} under the finite element approximation can be shown to be Gaussian with mean zero and covariance $\mv{C}$. Thus, one has
\begin{equation}\label{paperD:eq:w}
 \mv{w} \sim \pN\left(0,\mv{K}^{-1}\mv{C}\mv{K}^{-1}\right).
\end{equation}
For higher order $\alpha/2 \in \N$, the weak solution is obtained recursively.
If, for example, $\alpha = 4$ the solution to
$(\kappa^2 - \Delta)^2X_0 = \noise$ is obtained by solving
${(\kappa^2 -\Delta)X_0 = \tilde{X}}$, where $\tilde{X}$ is the solution for the case $\alpha = 2$. This results in replacing the matrix $\mv{K}$ with a matrix $\mv{K}_{\alpha}$ defined recursively as ${\mv{K}_{\alpha} =\mv{K}\mv{C}^{-1}\mv{K}_{\alpha-2}}$, where $\mv{K}_2 = \mv{K}$. For more details about these representations in the Gaussian case, see \cite{lindgren10}.

So far, we have not specified how the basis functions $\{\varphi_i\}$ should be chosen, but this choice will determine the quality of the approximation as well as some computational properties. If, for example, Daubechies wavelets are used as basis functions, the precision matrix (inverse covariance matrix) $\mv{Q}$ for the weights is a sparse matrix \citep{bolin09}, which facilitates the use of efficient sparse matrix techniques when using this model. \cite{lindgren10} used piecewise linear basis functions induced by triangulating the domain, and in this case $\mv{C}$ is a sparse matrix, but its inverse is dense. To obtain a sparse precision matrix in this case (which is needed for efficient GMRF computations), one can approximate $\mv{C}$ with a diagonal matrix $\tilde{\mv{C}}$ with elements $\tilde{\mv{C}}_{ii} = \int\varphi_i(\mv{s})\md\mv{s}$. To simplify the notation later, we denote the $i$th element on the diagonal by $\mv{a}_i$ as it is the area where $\varphi_i > \varphi_j$ for $j\neq i$. For more details on this approximation and the choice of basis functions, see \citet{bolin09}.

\subsubsection{Laplace noise}
For the Laplace case, one has $\dot{M} = \dot{\Lambda}$ in the weak formulation \eqref{paperD:eq:weak1}. Under the finite element approximation, the left-hand side can, as in the Gaussian case, be written as $\mv{K}_{\alpha}\mv{w}$. 
Using Theorem 1 in \cite{aberg08b}, the distribution of the right-hand side in the case of Laplace noise is given by the characteristic function
\begin{equation*}
\phi_{\Lambda}(\mv{u}) =  \exp\left(\tau\int_{\Omega}i\gamma \mv{\varphi}(\mv{s})^{\trsp}\mv{u}-\log\left(1-i\mu\mv{\varphi}(\mv{s})^{\trsp}\mv{u} + \frac{\sigma^2}{2}(\mv{\varphi}(\mv{s})^{\trsp}\mv{u})^2\right)\md\mv{s}\right),
\end{equation*}
where $\mv{\varphi}(\mv{s}) = \left(\varphi_1(\mv{s}),\ldots,\varphi_n(\mv{s})\right)^{\trsp}$. This representation is not very convenient for approximation and simulation of the model. Instead we will use a representation based on the series expansion \eqref{paperD:eq:lfieldseries} of $\Lambda$. However, for a moment, we turn to the more general setup of type-G processes to hint at how this technique could be applied also for this broader class of random fields.  

Recall that  a L\'{e}vy process is type G if its increments can be represented as a Gaussian variance mixture $V^{1/2}Z$ where $Z$ is a standard Gaussian variable and $V$ is a non-negative infinitely divisible random variable. Clearly, the Laplace fields are of type G as their increments are of the form $\Gamma^{1/2} Z$ where $\Gamma$ is a gamma variable. \cite{rosinski91} showed that every L\'{e}vy process of type G can be represented as a series expansion similar to the expansion \eqref{paperD:eq:lfieldseries} for the Laplace fields. This expansion also holds in $\R^d$, and for a compact domain $D\in \R^d$ it can be written as
\begin{equation*}
M(\mv{s}) = \sum_{k=1}^{\infty}G_k g(\gamma_k)^{\frac{1}{2}}\mv{1}(\mv{s} \geq \mv{s}_k),
\end{equation*}
where the function $g$ is the generalized inverse of the tail L\'{e}vy measure for $V$ and the other variables are the same as in the Laplace case \eqref{paperD:eq:lfieldseries}. Since $V$ is infinitely divisible, there exists a non-decreasing L\'evy process $V(\mathbf s)$ with increments distributed the same as $V$. This process has the series representation
\begin{equation}\label{paperD:eq:Vseries}
V(\mv{s}) = \sum_{k=1}^{\infty}g(\gamma_k)^{\frac{1}{2}}\mv{1}(\mv{s} \geq \mv{s}_k).
\end{equation}
Now, consider the integral of some basis function $\varphi_i$ with respect to $M$, which can be represented as
\begin{equation}\label{paperD:eq:intseries}
\int_D \varphi_i(\mv{s})M(\md\mv{s}) \overset{d}{=} \sum_{k=1}^{\infty}\varphi_i(\mv{s}_k)G_k\sqrt{g(\gamma_k)}.
\end{equation}
Thus, the distribution of $(\int_D \varphi_1(\mathbf s) M(d\mathbf s),\dots,\int_D \varphi_n(\mathbf s) M(d\mathbf s))$ can be approximated in distribution by taking partial sums of the series in \eqref{paperD:eq:intseries}.
Another way of calculating the distribution is to evaluate the integrals by conditioning on the variance process $V(\mv{s})$ \citep{wiktorsson02}; given that $\int_D \varphi_i^2(\mv{s})V(\md\mv{s}) < \infty$, the integral conditionally on $V$ is simply a Gaussian variable 
\begin{equation*}
\int_D \varphi_i(\mv{s})M(\md\mv{s}) | V \sim \pN\left(0, \int_D \varphi_i^2(\mv{s})V(\md\mv{s})\right).
\end{equation*}

Going back to the case of Laplace noise. If $M$ is a Laplace field corresponding to the Laplace measure $\Lambda$, the variance process is a gamma process, $\Gamma(\mv{s})$, so by the argument above one has that the right hand side of \eqref{paperD:eq:weak1} under the finite element approximation and conditionally on the gamma process is $\pN(\tilde{\mv{m}},\tilde{\mv{\Sigma}})$, where the elements of $\tilde{\mv{m}}$ and $\tilde{\mv{\Sigma}}$ are given by
\begin{align*}\label{paperD:eq:lcov}
\tilde{\mv{\Sigma}}_{ij} &= \pC\left( \int_D \varphi_i(\mv{s})\Lambda(\md\mv{s})\,, \int_D \varphi_j(\mv{s})\Lambda(\md\mv{s})  \middle|\, \Gamma\right) =  \int_{D} \varphi_i(\mv{s})\varphi_j(\mv{s})\Gamma(\md\mv{s}),\\
\tilde{\mv{m}}_i &= \pE\left(\int_D \varphi_i(\mv{s})\Lambda(\md\mv{s}) \middle| \,\Gamma\right) = \gamma\int_D \varphi_i(\mv{s})\md\mv{s} + \int_D \varphi_i(\mv{s})\Gamma(\md\mv{s}).
\end{align*}
Given this, the weights $\mv{w}$ can be calculated conditionally on the gamma process, $\Gamma(\mv{s})$, as
\begin{equation}\label{paperD:eq:lw}
 \mv{w} | \Gamma\sim \pN\left(\mv{K}_{\alpha}^{-1}\tilde{\mv{m}},\mv{K}_{\alpha}^{-1}\tilde{\mv{\Sigma}}\mv{K}_{\alpha}^{-1}\right),
\end{equation}
where $\mv{K}_{\alpha}$ is defined recursively as in the Gaussian case. 

It would seem as one has not gained much by using the conditional representation since the conditional mean and covariances, $\tilde{\mv{m}}_i$ and $\tilde{\mv{\Sigma}}_{ij}$, do not have any simple distributions. One way of approximating them is to approximate the integrals with respect to the Gamma process using the right hand side of \eqref{paperD:eq:Vseries} with a finite number of terms. However, by using compactly supported linear basis functions, one can simplify things further. Thus, now assume that the basis functions are piecewise linear functions induced by some triangulation of the domain. One can then perform the same Markov approximation as in the Gaussian case. This results in an approximation of the right-hand side of \eqref{paperD:eq:weak1} conditionally on the gamma process distributed as $\pN(\mv{m},\mv{\Sigma})$ with \mbox{$\mv{m} = \gamma\tau\mv{a} + \mu\mv{\Gamma}$}, and $\mv{\Sigma} =  \diag(\mv{\Gamma})$. Here, the gamma variables $\mv{\Gamma}_i \sim \Gamma(\tau\mv{a}_i,1)$ are independent and $\mv{a}_i = \int \varphi_i(\mv{s})\md\mv{s}$, and these can be calculated without numerically estimating the integrals with respect to the gamma process. 

\cite{bolin09} studies how this approximation affects the resulting covariance function of the process in the Gaussian case, and it is shown that the error is small if the approximation is used for piecewise linear basis functions. Although additional studies are needed in the non-Gaussian case, the results are likely similar so that the simplification has no large impact on the approximation. Figures \ref{paperD:fig:simulation}-\ref{paperD:fig:2dcovs} show that the approximation is accurate in one and two dimensions as explained in Section \ref{paperD:sec:sampel}.

\subsection{The solution for general $\alpha>d/2$}
If one could approximate the solution to \eqref{paperD:eq:lspde} for $\alpha=1$, the recursive scheme discussed above could be used to represent the solutions for all positive odd $\alpha$. In the Gaussian case, \cite{lindgren10} use a least-squares method where the test functions are chosen as $\psi_i = (\kappa^2-\Delta)^{\frac{1}{2}}\varphi_i$. The left-hand side of \eqref{paperD:eq:weak1} can then be expressed as $\mv{K}\mv{w}$ and the right-hand side is a mean zero Gaussian variable with covariance matrix $\mv{K}$. This follows from Lemma 2 in \cite{lindgren10}, which shows that the covariance between element $i$ and element $j$ on the right-hand side can be written as
\begin{equation*}
\mv{\Sigma}_{ij} = \scal{(\kappa^2-\Delta)^{\frac{1}{2}}\varphi_i}{(\kappa^2-\Delta)^{\frac{1}{2}}\varphi_j} = \scal{(\kappa^2-\Delta)\varphi_i}{\varphi_j} = \mv{K}_{ij}.
\end{equation*}
The stochastic weights therefore form a GMRF $\mv{w}\sim \pN(\mv{0},\mv{K}^{-1})$.
This argument is unfortunately not applicable in the non-Gaussian case as the covariance between the elements given the gamma process $\Gamma(\mv{s})$ is
\begin{align*}
\mv{\Sigma}_{ij} &= \pC\left(\int_D (\kappa^2-\Delta)^{\frac{1}{2}}\varphi_i(\mv{s})\Lambda(\md\mv{s}),\int_D (\kappa^2-\Delta)^{\frac{1}{2}}\varphi_j(\mv{s})\Lambda(\md\mv{s}) \middle| \,\Gamma\right)\\
&=  \int_{D} \left((\kappa^2-\Delta)^{\frac{1}{2}}\varphi_i(\mv{s})\right)\left((\kappa^2-\Delta)^{\frac{1}{2}}\varphi_j(\mv{s})\right)\Gamma(\md\mv{s})\\
&\neq  \int_{D} \left((\kappa^2-\Delta)\varphi_i(\mv{s})\right)\varphi_j(\mv{s})\Gamma(\md\mv{s}).
\end{align*}
We have not been able to find an easy way of evaluating $\mv{\Sigma}_{ij}$ in the non-Gaussian case, and it seems as this least-squares procedure is not extendable to the non-Gaussian case. However, if one instead uses $\psi_i=\varphi_i$, the right-hand side of \eqref{paperD:eq:weak1} conditionally on the variance process is $\pN(\mv{m},\mv{\Sigma})$, as in the case $\alpha = 2$. With this as a starting point, one can use a finite element matrix transfer technique (FE-MTT) to obtain a discretized approximation of the solution. \cite{simpson08} studied such methods for sampling generalized Mat\'{e}rn fields on locally planar Riemannian manifolds, and argued that one could sample the stochastic weights for a general $\alpha$ using the matrix transfer equation
$(\mv{C}^{-1}\mv{K})^{\alpha/2}\mv{w} \sim \pN(0,\mv{C}^{-1})$.
To simplify the notations in later sections, denote $\mv{K}_{\alpha}=(\mv{C}^{-1}\mv{K})^{\alpha/2}$ and note that we now have changed the definition of $\mv{K}_{\alpha}$ from the one that was used for even $\alpha$. The weights $\mv{w}$ are then mean zero Gaussian with a precision matrix ${\mv{Q}_{\alpha} = \mv{K}_{\alpha}\mv{C}^{-1}\mv{K}_{\alpha}}$.
In the case $\alpha=2$, this discretization coincides with the approximation described above, but it can be used for any $\alpha>d/2$. 

Now in the non-Gaussian case, the results from the case $\alpha=2$ can be used directly to get a right-hand side that is Gaussian with mean $\mv{m}$ and covariance $\mv{\Sigma}$ conditionally on the variance process. As in the Gaussian case, this should be multiplied with $\mv{C}^{-1}$ to get consistency in the FE-MTT procedure. Hence, in the case of Laplace noise the weights are given by
\begin{equation}\label{paperD:eq:lwodd}
\mv{w} | \Gamma \sim \pN\left(\mv{K}_{\alpha}^{-1}\mv{C}^{-1}\mv{m},\mv{K}_{\alpha}^{-1}\mv{C}^{-1}\mv{\Sigma}\mv{C}^{-1}\mv{K}_{\alpha}^{-1}\right).
\end{equation}
Again, for the case $\alpha=2$, this coincides with the procedure described in the section above, and because of this we will from now on use this FE-MTT procedure for all $\alpha>d/2$. Consistency of the FE-MTT procedure follows from similar arguments as in \cite{simpson08}. These arguments do not provide a rate of convergence as the number of basis functions are increased, and as for the Gaussian case, the rate of convergence and the numerical properties of the approximation are strongly dependent on $\alpha$.

\section{Sampling from the model}\label{paperD:sec:sampel}
Using the finite element representation obtained in the previous section it is easy to generate samples from the SPDE \eqref{paperD:eq:lspde}. Assume that we want sample the model at locations $\mv{s} = \left(\mv{s}_1,\ldots, \mv{s}_n\right)$, and let $\mv{\Phi}$ be a matrix with elements $\mv{\Phi}_{ij} = \varphi_j(\mv{s}_i)$. Samples can now be generated using the following three-step algorithm.

\begin{alg}\label{paperD:alg:sampel}
Sampling the Laplace driven SPDE \eqref{paperD:eq:lspde}.
\begin{enumerate}
\item Generate two independent random vectors $\mv{\Gamma}$ and $\mv{Z}$, where $\mv{\Gamma}_i \sim \Gamma(\tau\mv{a}_i,1)$ and $\mv{Z}_i \sim \pN(0,1)$.
\item Let $\mv{\Lambda} = \gamma\tau\mv{a} + \mu\mv{\Gamma} + \diag(\sqrt{\mv{\Gamma}})\mv{Z}$ and calculate $\mv{w} = \mv{C}^{-1}\mv{\Lambda}$.
\item $\mv{X} = \mv{\Phi} \mv{K}_{\alpha}^{-1}\mv{w}$ is now a sample of the random field at the locations $\mv{s}$.
\end{enumerate}
\end{alg}
The last step could potentionally be computationally expensive for large simulations. However, if $\alpha$ is even, one can take advantage of the sparsity of $\mv{K}_{\alpha}$ and solve the equation system $\mv{v} = \mv{K}_{\alpha}^{-1}\mv{w}$ efficiently without calculating the inverse by using Cholesky factorization and back substitution as suggested by \cite{rue1}. For other $\alpha>d/2$, $\mv{K}_{\alpha}$ is not sparse and the Cholesky method will not improve the computational efficiency. However, as \cite{simpson08} shows, one can instead use Krylov subspace methods in the calculations to obtain efficient sampling schemes. The basic problem for general $\alpha$ is to solve the matrix equation $\mv{v} = (\mv{C}^{-1}\mv{K})^{-\frac{\alpha}{2}}\mv{w}$, and there are a number of methods with different computational properties that can be used. In this work we use the method by \cite{trefethen08}, which is based on combining contour integrals evaluated by the periodic trapezoid rule with conformal maps involving Jacobi elliptic functions. 

\begin{figure}[t]
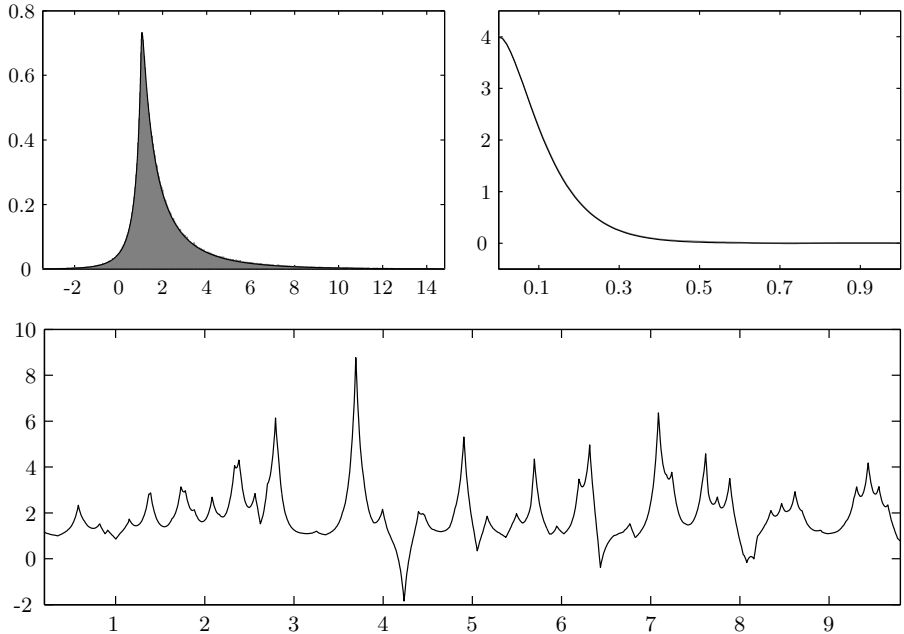

\begin{center}
\small
\resizebox{0.5\linewidth}{!}{\input{figs/pdf1d1.tex}}%
\resizebox{0.5\linewidth}{!}{\input{figs/cov1d1.tex}}\\%
\resizebox{\linewidth}{!}{\input{figs/sim1d1.tex}}%
\end{center}
\vspace{-0.5cm}
\caption{The lower panel shows a simulation of the Laplace driven SPDE \eqref{paperD:eq:lspde} on $\R$ with parameters $\mu=\gamma=\sigma=1$, $\tau=2$, $\kappa = 15$, and $\alpha = 2$. The upper left panel shows a histogram of the samples from $1000$ simulations together with the true density. The upper right panel shows the empirical covariance function for the samples (grey curve) together with the true Mat\'{e}rn covariance function (black curve). It is difficult to see the grey curve since the two curves are very similar.}
\label{paperD:fig:simulation}
\end{figure}

In Figure \ref{paperD:fig:simulation}, a simulation of a process on $\R$ with parameters ${\mu=\gamma=\sigma=1}$, $\tau=2$, $\kappa = 15$, and $\alpha = 2$ is shown. Since $\mv{K}_{\alpha}$ is sparse in this case, the Cholesky method is used for the simulation. In the upper left panel, a histogram of the samples from $1000$ simulations is shown together with the theoretical density, calculated using numerical Fourier inversion of the characteristic function~\eqref{paperD:eq:charf}. In the upper right panel, the empirical covariance function of the samples is shown together with the theoretical Mat\'{e}rn covariance function. Two more examples of densities and covariance functions for different parameter settings are shown in Figure \ref{paperD:fig:sim2}. In the upper panels, we have $\alpha=1$, which results in an exponential covariance function. The other parameters are $\mu=\gamma=0$, $\sigma=1$, and ${\tau=\kappa=10}$, which results in a symmetric distribution. In the lower panels, we have $\alpha=3.5$ which results in a smoother field. The other parameters are $\mu=\sigma=0.1$, $\gamma=0$, $\tau=10$, and $\kappa=20$, which results in an asymmetric distribution. In both cases in Figure \ref{paperD:fig:sim2}, the Krylov subspace method is used for the simulations.

In Figure \ref{paperD:fig:2dcovs} and Figure \ref{paperD:fig:2dcovs2}, two simulations of fields on $\R^2$ are shown together with the corresponding covariance functions, densities, and empirically estimated versions based on $1000$ simulations each. As seen in the figures for all five examples, there is a close agreement between the histograms and the true densities, and between the true covariance functions and the empirically estimated covariance functions for all these parameter settings, indicating that the approximation procedure works as intended. A more detailed analysis of the simulation procedure is outside the scope of this article, but it should be noted that the SPDE approximation using piecewise linear basis functions does not provide convergence of higher-order derivatives, and the simulation procedure is therefore not appropriate for applications where such properties are important.

\begin{figure}[t]
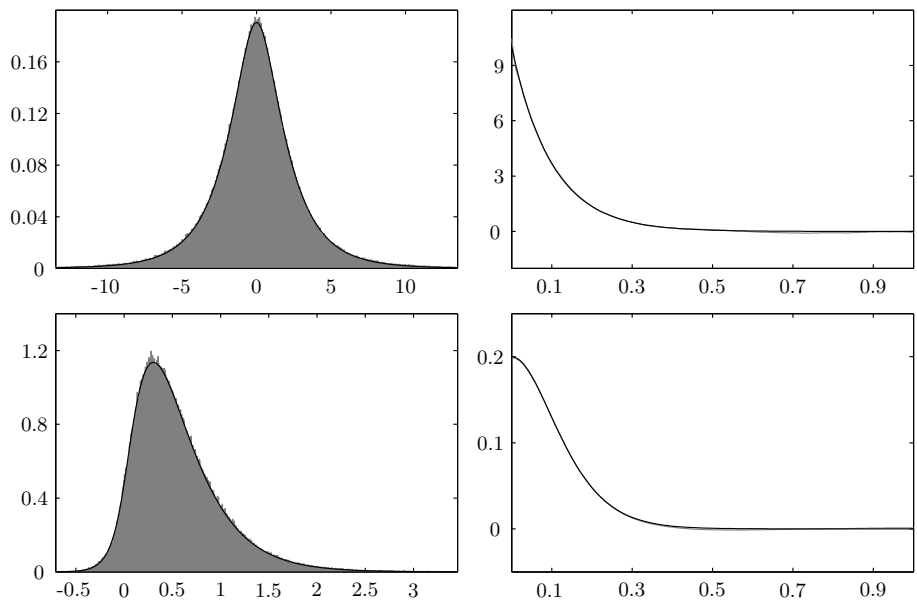

\begin{center}
\small
\resizebox{0.5\linewidth}{!}{\input{figs/pdf1d2.tex}}%
\resizebox{0.5\linewidth}{!}{\input{figs/cov1d2.tex}}\\%
\resizebox{0.5\linewidth}{!}{\input{figs/pdf1d3.tex}}%
\resizebox{0.5\linewidth}{!}{\input{figs/cov1d3.tex}}%
\end{center}
\vspace{-0.5cm}
\caption{Simulation results as in Figure \ref{paperD:fig:simulation} with different parameters. The top row shows a symmetric case with parameters $\mu=\gamma=0$, $\sigma=1$, $\kappa=\tau=10$, and $\alpha=1$. The bottom row shows an asymmetric case with parameters $\mu=\sigma=0.1$, $\gamma=0$, $\kappa=20$, $\tau=10$ and $\alpha=3.5$.}
\label{paperD:fig:sim2}
\end{figure}

\begin{figure}[t]
\begin{center}
\small
\resizebox{\linewidth}{!}{\includegraphics[bb=57 278 563 495,clip=]{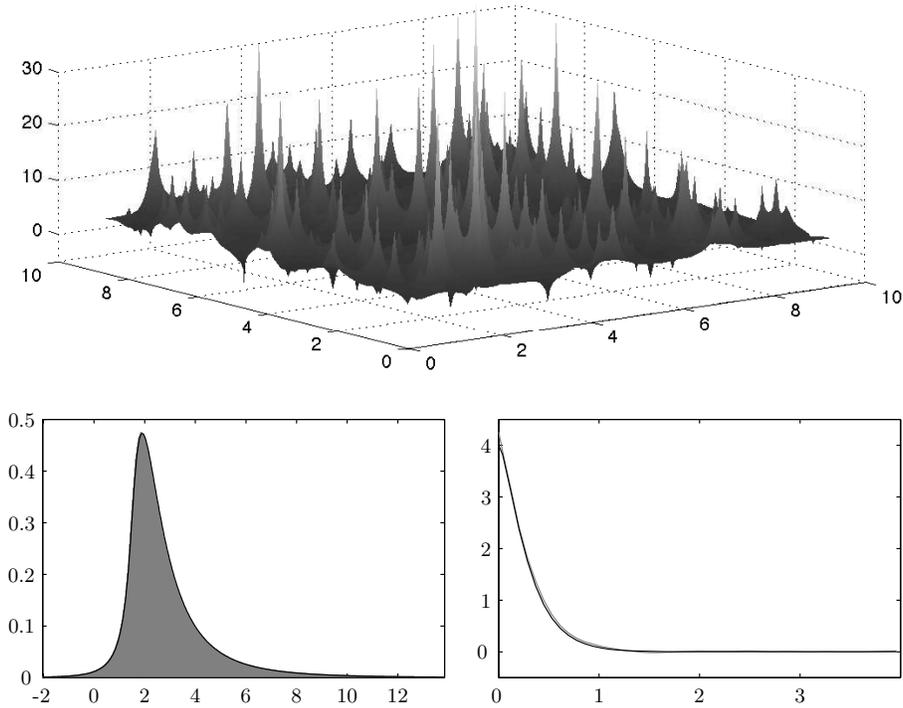}}\\
\small
\resizebox{0.5\linewidth}{!}{\input{figs/pdf2d1.tex}}%
\resizebox{0.5\linewidth}{!}{\input{figs/cov2d1.tex}}\\%

\end{center}
\vspace{-0.5cm}
\caption{
A simulation of an asymmetric model \eqref{paperD:eq:lspde} in $\R^2$ where the parameters are $\kappa = 5$, $\sigma=\mu=\gamma = 1$, $\tau=2$, and $\alpha=2$. The covariance functions and densities for these fields can be seen in the second row. The empirically estimated versions are based on $1000$ simulations.}
\label{paperD:fig:2dcovs}
\end{figure}

\begin{figure}[t]
\begin{center}
\small
\resizebox{\linewidth}{!}{\includegraphics[bb=57 278 563 495,clip=]{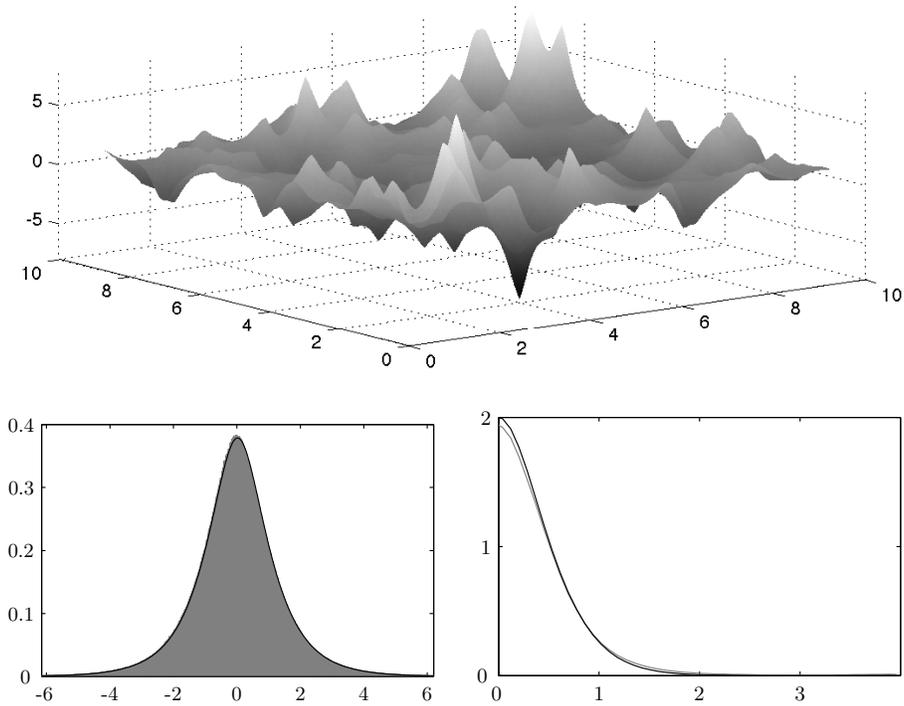}}
\resizebox{0.5\linewidth}{!}{\input{figs/pdf2d2.tex}}%
\resizebox{0.5\linewidth}{!}{\input{figs/cov2d2.tex}}%
\end{center}
\vspace{-0.5cm}
\caption{
A simulation of a symmetric model \eqref{paperD:eq:lspde} in $\R^2$ with parameters $\kappa = 5$, $\sigma=1$, $\mu=\gamma = 0$, $\tau=2$, and $\alpha=4$. The covariance function and density are shown in the second row. The empirically estimated versions are based on $1000$ simulations.}
\label{paperD:fig:2dcovs2}
\end{figure}

\section{Parameter estimation}\label{paperD:sec:estimation}
Parameter estimation for Laplace moving average models is not easy since there is no closed form expression for the parameter likelihood. Recently, \cite{joerg10} derived a method of moments-based estimation procedure for these types of models. In their method, the convolution kernel is first estimated from the spectral density of the data, and given the estimated kernel, the parameters in the Laplace distribution are estimated by fitting the theoretical moments of the Laplace distribution to the sample moments. The method is quite simple although some special care has to be taken to handle the cases when the method of moments equation system does not have a solution, which can happen for certain values of the sample skewness and excess kurtosis.

Using the SPDE formulation, parameter estimation can instead be performed in a likelihood framework. One of the advantages with this is that maximum likelihood parameter estimates always are in the allowed parameter space. Another advantage is that the estimates will account for all relevant information in the data, which might not be the case for method of moment estimates. 

To be able to estimate the parameters in a maximum likelihood framework, the problem is interpreted as a missing data problem which facilitates use of the Expectation Maximization (EM) algorithm \citep{Dempster77}. The proposed EM algorithm is based on the same ideas as the ones in \citet{Lange89} and \citet{Protassov04} which looked at EM estimation in the case of iid observations of certain Gaussian mixtures. Our main contribution is the extension of these ideas to the random field setting.

Assume we have measurements $\mv{X}$ of the process $X(\mv{s})$ taken at some locations and that the Hilbert space approximation procedure is used with a basis obtained by triangulating the measurement locations. In this case, the matrix $\mv{\Phi}$ is diagonal and conditioning on the measurements and the parameters is equivalent to conditioning on $\mv{\Lambda}$ and the parameters as there is a one-to-one correspondence between the two through $\mv{\Lambda} = \mv{C}\mv{K}_{\alpha}\mv{\Phi}^{-1}\mv{X}$, see Algorithm \ref{paperD:alg:sampel}. To obtain simpler updating expressions, we first make a change of variables by introducing the parameter $\bar{\gamma} = \gamma\tau$ and estimate this parameter instead of $\gamma$. As for Gaussian Mat\'{e}rn models, the shape parameter $\nu$ is difficult to estimate accurately and it is therefore assumed to be known throughout this section and no attempt is made at estimating it.

Augmenting the data with the unknown (missing) gamma variables, the augmented likelihood is $L(\mv{\theta}|\mv{X},\mv{\Gamma}) = \pi(\mv{X}|\mv{\Gamma},\mv{\theta})\pi(\mv{\Gamma}|\mv{\theta})$,
and the loss-function that is needed for the EM-procedure is
\begin{equation*}
\mathcal{Q}(\mv{\theta},\mv{\theta}^{(j)}) = \pE\left(\log L(\mv{\theta}|\mv{X},\mv{\Gamma}) \middle|\mv{X},\mv{\theta}^{(j)}\right),
\end{equation*}
where $\mv{\theta}^{(j)}$ is an estimate of $\mv{\theta} = \left(\kappa,\sigma,\mu,\bar{\gamma},\tau\right)$ at iteration $j$, and the expectation is taken according to the distribution of $\mv{\Gamma}$ given $\mv{X}$. We have ${\mv{X}|\mv{\Gamma},\mv{\theta} \sim \pN(\mv{m},\sigma^2\mv{\Sigma})}$, where 
 $\mv{m}=\mv{\Phi}\mv{K}_{\alpha}^{-1}\mv{C}(\bar{\gamma}\mv{a} + \mu\mv{\Gamma})$, $\mv{\Sigma} = \mv{\Phi}\mv{K}_{\alpha}^{-1}\mv{C}\mv{D}_{\Gamma}\mv{C}\mv{K}_{\alpha}^{-1}\mv{\Phi}$, and $\mv{D}_{\Gamma}$ is the diagonal matrix with the vector $\mv{\Gamma}$ on the main diagonal. The second part of the augmented likelihood can be written as $\pi(\mv{\Gamma}|\mv{\theta}) = \prod \pi(\mv{\Gamma}_i|\mv{\theta})$ since the components in $\mv{\Gamma}$ are independent gamma variables $\mv{\Gamma}_i$ with shape parameters $\tau \mv{a}_i$ and scale one, where $\mv{a}_i$ are known constants depending on the basis used. The log-likelihood is
\begin{align*}
\log L(\mv{\theta}|\mv{X},\mv{\Gamma}) = & -n\log(\sigma) + \log(|\mv{K}_{\alpha}|) - \frac{1}{2\sigma^2}(\mv{X}-\mv{m})^{\trsp}\mv{\Sigma}^{-1}(\mv{X}-\mv{m})\\ &+\sum_{i=1}^n\left(\tau\mv{a}_i\log\mv{\Gamma}_i -\log\Gamma(\tau\mv{a}_i)\right) + C,
\end{align*}
where the constant $C$ does not depend on the unknown parameters. Thus, using the relation between $\mv{X}$ and $\mv{\Lambda}$, the loss-function is
\begin{equation*}
\begin{split}
\mathcal{Q}(\mv{\theta},\mv{\theta}^{(j)}) =& -n\log(\sigma) + \log(|\mv{K}_{\alpha}|) - \frac{1}{2\sigma^2}\left((\mv{\Lambda}-\bar{\gamma} \mv{a})^{\trsp}\mv{D}_{\pE(\mv{\Gamma}^{-1}|\star)}(\mv{\Lambda}-\bar{\gamma} \mv{a}) \right.\\ &+\left.\mu^2 \mv{1}^{\trsp}\pE(\mv{\Gamma}|\star)
+ 2\bar{\gamma}\mu \mv{a}^{\trsp}\mv{1} - 2\mu \mv{\Lambda}^{\trsp}\mv{1}\right) \\ &+\sum_{i=1}^n\left(\tau\mv{a}_i\pE(\log\mv{\Gamma}_i|\star) -\log\Gamma(\tau\mv{a}_i)\right) + C,
\end{split}
\end{equation*}
where $\pE(\cdot|\star)$ denotes  $\pE(\cdot|\theta^{(j)},\mv{X})$. The expectations needed to evaluate the loss-function are $\pE(\mv{\Gamma}|\star)$, $\pE(\mv{\Gamma}^{-1}|\star)$, and $\pE(\log\mv{\Gamma}_i|\star)$. To calculate these, first note that \cite[see][formula 3.472.9]{GR00}
\begin{equation}\label{paperD:eq:macdonald}
\mathcal{I}(a,b,c) = \int_0^\infty x^{a-1}e^{-\frac{b}{x}-cx}\md x = 2\left(\frac{b}{c}\right)^{\frac{a}{2}}K_{a}\left(2\sqrt{bc}\right).
\end{equation}
Using this expression, the expectation $\pE(\mv{\Gamma}_i|\star)$ can be written as
\begin{align*}
\hspace{-0.4cm}\pE(\mv{\Gamma}_i|\star) &= \int \mv{\Gamma}_i \pi(\mv{\Gamma}_i | \mv{X},\mv{\theta})\md \mv{\Gamma}_i = 
\frac{\int \mv{\Gamma}_i \pi(\mv{X} | \mv{\Gamma}_i,\mv{\theta}) \pi(\mv{\Gamma}_i|\mv{\theta})\md \mv{\Gamma}_i}{\pi(\mv{X}|\mv{\theta})}\\
&=\frac{\int \mv{\Gamma}_i \pi(\mv{X} | \mv{\Gamma}_i,\mv{\theta}) \pi(\mv{\Gamma}_i|\mv{\theta})\md \mv{\Gamma}_i}{\int \pi(\mv{X}|\mv{\Gamma}_i,\mv{\theta})\pi(\mv{\Gamma}_i|\mv{\theta}) \md \mv{\Gamma}_i} = \frac{\mathcal{I}\left(\tau\mv{a}_i+\frac{1}{2},\frac{(\mv{\Lambda}_i-\bar{\gamma}\mv{a}_i)^2}{2\sigma^2},1+\frac{\mu^2}{2\sigma^2}\right)}{\mathcal{I}\left(\tau\mv{a}_i-\frac{1}{2},\frac{(\mv{\Lambda}_i-\gamma \tau\mv{a}_i)^2}{2\sigma^2},1+\frac{\mu^2}{2\sigma^2}\right)} \\
&= \frac{|\mv{\Lambda}_i-\bar{\gamma}\mv{a}_i|}{\sqrt{2\sigma^2+\mu^2}} \frac{K_{\tau\mv{a}_i+\frac{1}{2}}\left(\sigma^{-2}|\mv{\Lambda}_i-\bar{\gamma}\mv{a}_i|\sqrt{2\sigma^2+\mu^2}\right)} {K_{\tau\mv{a}_i-\frac{1}{2}}\left(\sigma^{-2}|\mv{\Lambda}_i-\bar{\gamma}\mv{a}_i|\sqrt{2\sigma^2+\mu^2}\right)}.
\end{align*}
If the argument in the Bessel functions is very small or very large one might get numerical problems when evaluating this expression depending on how it is implemented. In the case of small arguments, one can use the following approximation to improve the numerical stability
\begin{align*}
&K_{a}(x) \approx \frac{\Gamma(|a|)}{2}\left(\frac{2}{x}\right)^{|a|}, \mbox{ if $a \neq 0$ and $x \ll \sqrt{|a|+1}$}.
\end{align*}
The expectation $\pE(\mv{\Gamma}_i|\star)$ then simplifies to
\begin{align*}
\pE(\mv{\Gamma}_i|\star) &\approx 
\begin{cases}
\left(\tau\mv{a}_i - \frac{1}{2}\right)\frac{2\sigma^2}{2\sigma^2 + \mu^2}, & \tau > \frac{1}{2\mv{a}_i},\\
\frac{\Gamma\left(\tau\mv{a}_i+\frac{1}{2}\right)}{\Gamma\left(\frac{1}{2}-\tau\mv{a}_i\right)}\frac{\left(2\sigma^2\right)^{2\tau\mv{a}_i}}{\left(2\sigma^2 + \mu^2\right)^{\frac{2\tau\mv{a}_i+1}{2}}}|\mv{\Lambda}_i-\bar{\gamma}\mv{a}_i|^{1-2\tau\mv{a}_i}, & \tau < \frac{1}{2\mv{a}_i}.
\end{cases}
\end{align*}
In the case of large arguments, one can instead use the approximation
\begin{align*}
&\frac{K_{a}(x)}{K_{a-1}(x)} \approx 1 + \left(a - \frac{1}{2}\right)\frac{1}{x},
\end{align*}
which gives the following approximation for $\pE(\mv{\Gamma}_i|\star)$
\begin{equation*}
\pE(\mv{\Gamma}_i|\star) \approx \frac{|\mv{\Lambda}_i-\bar{\gamma} \mv{a}_i|}{\sqrt{2\sigma^2+\mu^2}} + \frac{\tau\mv{a}_i\sigma^2}{2\sigma^2 + \mu^2}.
\end{equation*}

The expectation $\pE(\mv{\Gamma}_i^{-1}|\star)$ is calculated similarly using \eqref{paperD:eq:macdonald} and can be written as 
\begin{equation}\label{paperD:eq:Egammainv}
\pE(\mv{\Gamma}_i^{-1}|\star) = \frac{\sqrt{2\sigma^2+\mu^2}}{|\mv{\Lambda}_i-\bar{\gamma}\mv{a}_i|} \frac{K_{\tau\mv{a}_i-\frac{3}{2}}\left(\sigma^{-2}|\mv{\Lambda}_i-\bar{\gamma}\mv{a}_i|\sqrt{2\sigma^2+\mu^2}\right)} {K_{\tau\mv{a}_i-\frac{1}{2}}\left(\sigma^{-2}|\mv{\Lambda}_i-\bar{\gamma}\mv{a}_i|\sqrt{2\sigma^2+\mu^2}\right)}.
\end{equation}
Evaluating modified Bessel functions numerically is computationally expensive and should therefore be avoided as much as possible when implementing the estimation procedure. To that end, one can express $K_{\tau\mv{a}_i-\frac{3}{2}}(\cdot)$ using the following recurrence relationship for modified Bessel functions
\begin{equation*}
K_{a}(x) = K_{a+2}(x) - \frac{2(a + 1)}{x}K_{a+1}(x),
\end{equation*}
giving the following expression for $\pE(\mv{\Gamma}_i^{-1}|\star)$ in terms of $\pE(\mv{\Gamma}_i|\star)$
\begin{equation*}
\pE(\mv{\Gamma}_i^{-1}|\star) = \frac{(\mu^2 + 2\sigma^2)\pE(\mv{\Gamma}_i|\star) - \sigma^2(2\tau\mv{a}_i -1)}{(\mv{\Lambda}_i-\bar{\gamma}\mv{a}_i)^2}.
\end{equation*}
Using this expression instead of \eqref{paperD:eq:Egammainv}, one only has to evaluate two modified Bessel functions instead of three. 

Finally, the expectation $\pE(\log(\mv{\Gamma}_i)|\star)$ is similarly written as
\begin{equation*}
\pE(\log(\mv{\Gamma}_i)|\star) = 
\frac{\int \log(\mv{\Gamma}_i) \pi(\mv{X} | \mv{\Gamma}_i,\mv{\theta}) \pi(\mv{\Gamma}_i|\mv{\theta})\md \mv{\Gamma}_i}{\pi(\mv{X}|\mv{\theta})}.
\end{equation*}
The denominator is the same as in the previous expectations, while calculating the nominator requires evaluating an integral on the form
\begin{equation}\label{paperD:eq:logint}
\mathcal{I}_{log}(a,b,c) = \int_0^\infty \log(x)x^{a-1}\exp\left(-\frac{b}{x}-cx\right)\md x.
\end{equation}
To calculate this integral, we differentiate \eqref{paperD:eq:macdonald} with respect to $a$ and obtain
\begin{align*}
\mathcal{I}_{log}(a,b,c) &= \frac{\pd}{\pd a}\int_0^\infty x^{a-1}e^{-\frac{b}{x}-cx}\md x = \frac{\pd}{\pd a}\left(2\left(\frac{b}{c}\right)^{\frac{a}{2}}K_{a}\left(2\sqrt{bc}\right)\right)\\
&= 2\left(\frac{b}{c}\right)^{\frac{a}{2}}\left(\log\left(\frac{b}{c}\right)K_{a}\left(2\sqrt{bc}\right) + \frac{\pd}{\pd a}K_{a}\left(2\sqrt{bc}\right)\right).
\end{align*}
The derivative of $K_{a}(2\sqrt{bc})$ with respect to $a$ can be expressed using infinite sums of gamma- and polygamma functions; however, in this case it is easier to numerically approximate the derivative using for example forward differences:
\begin{equation*}
\frac{\pd}{\pd a}K_{a}\left(2\sqrt{bc}\right) \approx \frac{K_{a+\epsilon}\left(2\sqrt{bc}\right) - K_{a}\left(2\sqrt{bc}\right)}{\epsilon}.
\end{equation*}
Using this expression, we approximate $\pE(\log(\Gamma_i)|\star)$ as
\begin{align*}
\pE(\log(\mv{\Gamma}_i)|\star) = & \frac{\mathcal{I}_{log}\left(\tau\mv{a}_i-\frac{1}{2},\frac{(\mv{\Lambda}_i-\bar{\gamma} \mv{a}_i)^2}{2\sigma^2},1+\frac{\mu^2}{2\sigma^2}\right)}{\mathcal{I}\left(\tau\mv{a}_i-\frac{1}{2},\frac{(\mv{\Lambda}_i-\bar{\gamma} \mv{a}_i)^2}{2\sigma^2},1+\frac{\mu^2}{2\sigma^2}\right)}
\\
\approx & \log\left(\frac{|\mv{\Lambda}_i-\bar{\gamma} \mv{a}_i|}{\sqrt{\mu^2+2\sigma^2}}\right) -\frac{1}{\epsilon} \\ &+  \frac{1}{\epsilon}\frac{K_{\tau\mv{a}_i-\frac{1}{2}+\epsilon}\left(\sigma^{-2}|\mv{\Lambda}_i-\bar{\gamma} \mv{a}_i|\sqrt{2\sigma^2+\mu^2}\right)}{K_{\tau\mv{a}_i-\frac{1}{2}}\left(\sigma^{-2}|\mv{\Lambda}_i-\bar{\gamma} \mv{a}_i|\sqrt{2\sigma^2+\mu^2}\right)}.
\end{align*}

To obtain the updating equations for the parameters, the loss-function should be maximized with respect to each parameter, for example by differentiating it with respect to the parameters and setting the derivatives equal to zero. Since the system of equations obtained from this procedure is not analytically solvable, one would have to iterate numerically in each step to obtain the parameter updates if the EM algorithm is used without modifications. A better alternative is to use an Expectation Conditional Maximization (ECM) algorithm \citep{mengrubin93} where the M-step is divided into two conditional maximization steps. In the first step, the parameters of the Laplace noise is updated conditionally on the current value of $\kappa$, and in the second step $\kappa$ is updated conditionally on the other parameters. Differentiating the loss-function with respect to $\mu$, $\bar{\gamma}$, and $\sigma$ and setting the derivatives equal to zero yields the following updating rules
\begin{align*}
\mu^{(j+1)} &= \frac{(\mv{\Lambda}^{\trsp}\mv{1})(\mv{a}^{\trsp}\mv{D}_{\pE(\mv{\Gamma}^{-1}|\star)}\mv{a})-(\mv{a}^{\trsp}\mv{1})(\mv{\Lambda}^{\trsp}\mv{D}_{\pE(\mv{\Gamma}^{-1}|\star)}\mv{a})}{(\mv{1}^{\trsp}\pE(\mv{\Gamma}|\star))(\mv{a}^{\trsp}\mv{D}_{\pE(\mv{\Gamma}^{-1}|\star)}\mv{a})-(\mv{1}^{\trsp}\mv{a})^2}, \\
\bar{\gamma}^{(j+1)} &= \frac{(\mv{1}^{\trsp}\pE(\mv{\Gamma}|\star))(\mv{\Lambda}^{\trsp}\mv{D}_{\pE(\mv{\Gamma}^{-1}|\star)}\mv{a})-(\mv{\Lambda}^{\trsp}\mv{1})(\mv{a}^{\trsp}\mv{1})}{(\mv{1}^{\trsp}\pE(\mv{\Gamma}|\star))(\mv{a}^{\trsp}\mv{D}_{\pE(\mv{\Gamma}^{-1}|\star)}\mv{a})-(\mv{1}^{\trsp}\mv{a})^2}, \\ 
\sigma^{(j+1)} &= \frac{1}{\sqrt{n}}\left(\mv{\Lambda}^{\trsp}\mv{D}_{\pE(\mv{\Gamma}^{-1}|\star)}\mv{\Lambda}+2\frac{(\mv{\Lambda}^{\trsp}\mv{D}_{\pE(\mv{\Gamma}^{-1}|\star)}\mv{a})(\mv{\Lambda}^{\trsp}\mv{1})(\mv{1}^{\trsp}\mv{a})}{(\mv{1}^{\trsp}\pE(\mv{\Gamma}|\star))(\mv{a}^{\trsp}\mv{D}_{\pE(\mv{\Gamma}^{-1}|\star)}\mv{a})-(\mv{1}^{\trsp}\mv{a})^2}\right.\\
&\left.
-\frac{(\mv{\Lambda}^{\trsp}\mv{D}_{\pE(\mv{\Gamma}^{-1}|\star)}\mv{a})^2 (\mv{1}^{\trsp}\pE(\mv{\Gamma}|\star)) + (\mv{a}^{\trsp}\mv{D}_{\pE(\mv{\Gamma}^{-1}|\star)}\mv{a})(\mv{\Lambda}^{\trsp}\mv{1})^2}{(\mv{1}^{\trsp}\pE(\mv{\Gamma}|\star))(\mv{a}^{\trsp}\mv{D}_{\pE(\mv{\Gamma}^{-1}|\star)}\mv{a})-(\mv{1}^{\trsp}\mv{a})^2}
\right)^{\frac{1}{2}}.
\end{align*}
In general, there is no closed form expression for the conditional updating equation for $\tau$, so the following equation is maximized numerically to obtain $\tau^{(j+1)}$ 
\begin{equation*}
\mathcal{Q}_{\tau}= \sum_{i=1}^n\left(\tau\mv{a}_i\pE(\log\mv{\Gamma}_i|\star) -\log\Gamma(\tau\mv{a}_i)\right).
\end{equation*}
In the special case when all $\mv{a}_i$ are equal to some value $a$, which for example is the case if a triangulation induced by a regular lattice is used in the Hilbert space approximation, the solution can be written as
\begin{equation*}
\tau^{(j+1)} = \frac{1}{a}\psi^{-1}\left(\frac{1}{n}\sum_{i=1}^n \pE(\log\mv{\Gamma}_i|\star)\right),
\end{equation*}
where $\psi^{-1}(\cdot)$ is the inverse of the digamma function.
Finally $\kappa$ is updated conditionally on the other parameters.
There is no closed form expression for the updating equation for $\kappa$ either, so the following expression is maximized numerically with respect to $\kappa$,
\begin{align*}
\mathcal{Q}_\kappa =& \log(|\mv{K}_{\alpha}|)- \frac{1}{2(\sigma^{(i+1)})^2}\left(\mv{\Lambda}^{\trsp}\mv{D}_{\pE(\Gamma^{-1}|\star)}\mv{\Lambda} \right. \\ &- \left. 2\bar{\gamma}^{(j+1)}\mv{\Lambda}^{\trsp}\mv{D}_{\pE(\Gamma^{-1}|\star)}\mv{a}- 2\mu^{(j+1)} \mv{\Lambda}^{\trsp}\mv{1}\right).
\end{align*}
By the construction of $\mv{K}_{\alpha}$, its log-determinant can be written as
\begin{equation*}
\log(|\mv{K}_{\alpha}|) = \frac{\alpha}{2}\log|\mv{C}^{-1}\mv{G}+\kappa^2\mv{I}| = \frac{\alpha}{2}\sum_{i=1}^n \log(\lambda_i + \kappa^2),
\end{equation*}
where $\lambda_i$ denotes the $i$th eigenvalue of $\mv{C}^{-1}\mv{G}$. If the size of $\mv{K}_{\alpha}$ is small, these eigenvalues can be pre-calculated as they do not depend on the parameters. For larger problems is it most efficient to calculate the log-determinant in each iteration using a sparse Cholesky factorization of $\mv{K}= \mv{G}+\kappa^2\mv{C}$. 

As shown by \cite{mengrubin93}, the ECM algorithm has the same convergence properties as the ordinary EM algorithm. The likelihood is increasing for each iteration and the convergence is linear. Hence, we do not lose any rate of convergence by using the ECM algorithm instead of the EM algorithm.

\section{A simulation study}\label{paperD:sec:study}

\begin{figure}[t]
\begin{center}
\small
\resizebox{0.5\linewidth}{!}{
%
%
\begin{psfrags}%
\psfragscanon%
%
\psfrag{s08}[b][b]{\setlength{\tabcolsep}{0pt}\begin{tabular}{c}a\end{tabular}}%
\psfrag{s10}[][]{\setlength{\tabcolsep}{0pt}\begin{tabular}{c} \end{tabular}}%
\psfrag{s11}[][]{\setlength{\tabcolsep}{0pt}\begin{tabular}{c} \end{tabular}}%
\psfrag{s12}[l][l]{E}%
\psfrag{s13}[l][l]{A}%
\psfrag{s14}[l][l]{C}%
\psfrag{s15}[l][l]{E}%
%
\psfrag{x01}[t][t]{-2}%
\psfrag{x02}[t][t]{0}%
\psfrag{x03}[t][t]{2}%
%
\psfrag{v01}[r][r]{0}%
\psfrag{v02}[r][r]{0.5}%
\psfrag{v03}[r][r]{1}%
\psfrag{v04}[r][r]{1.5}%
\psfrag{v05}[r][r]{2}%
%
\resizebox{6cm}{!}{\includegraphics{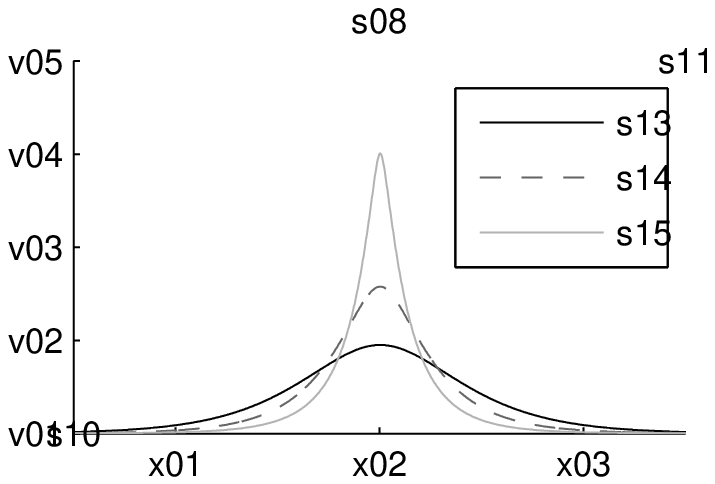}}%
\end{psfrags}%
%
}%
\resizebox{0.5\linewidth}{!}{
%
%
\begin{psfrags}%
\psfragscanon%
%
\psfrag{s08}[b][b]{\setlength{\tabcolsep}{0pt}\begin{tabular}{c}b\end{tabular}}%
\psfrag{s10}[][]{\setlength{\tabcolsep}{0pt}\begin{tabular}{c} \end{tabular}}%
\psfrag{s11}[][]{\setlength{\tabcolsep}{0pt}\begin{tabular}{c} \end{tabular}}%
\psfrag{s12}[l][l]{F}%
\psfrag{s13}[l][l]{B}%
\psfrag{s14}[l][l]{D}%
\psfrag{s15}[l][l]{F}%
%
\psfrag{x01}[t][t]{-2}%
\psfrag{x02}[t][t]{0}%
\psfrag{x03}[t][t]{2}%
\psfrag{x04}[t][t]{4}%
%
\psfrag{v01}[r][r]{0}%
\psfrag{v02}[r][r]{0.5}%
\psfrag{v03}[r][r]{1}%
\psfrag{v04}[r][r]{1.5}%
%
\resizebox{6cm}{!}{\includegraphics{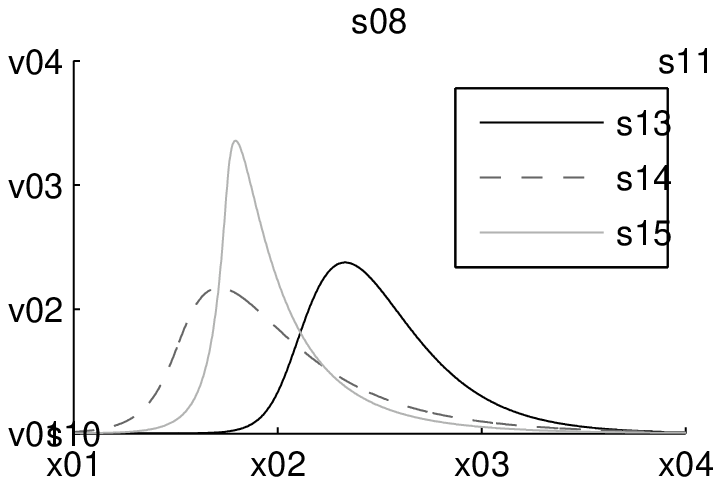}}%
\end{psfrags}%
%
}\\%
\resizebox{0.5\linewidth}{!}{
%
%
\begin{psfrags}%
\psfragscanon%
%
\psfrag{s08}[b][b]{\setlength{\tabcolsep}{0pt}\begin{tabular}{c}c\end{tabular}}%
\psfrag{s10}[][]{\setlength{\tabcolsep}{0pt}\begin{tabular}{c} \end{tabular}}%
\psfrag{s11}[][]{\setlength{\tabcolsep}{0pt}\begin{tabular}{c} \end{tabular}}%
\psfrag{s12}[l][l]{K}%
\psfrag{s13}[l][l]{G}%
\psfrag{s14}[l][l]{I}%
\psfrag{s15}[l][l]{K}%
%
\psfrag{x01}[t][t]{-2}%
\psfrag{x02}[t][t]{0}%
\psfrag{x03}[t][t]{2}%
%
\psfrag{v01}[r][r]{0}%
\psfrag{v02}[r][r]{0.5}%
\psfrag{v03}[r][r]{1}%
\psfrag{v04}[r][r]{1.5}%
%
\resizebox{6cm}{!}{\includegraphics{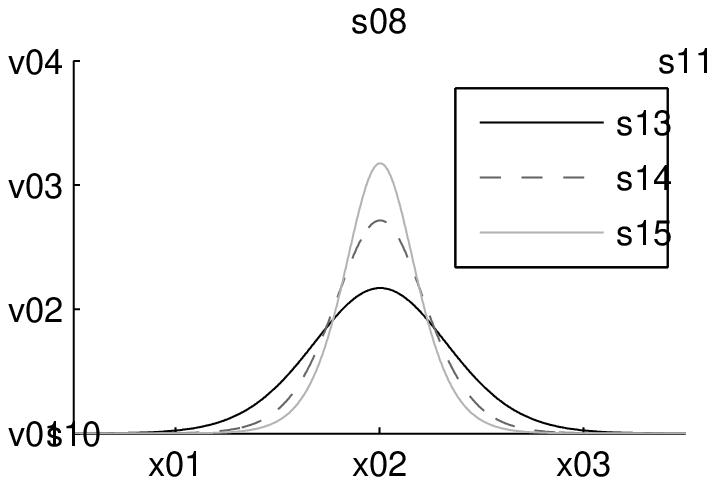}}%
\end{psfrags}%
%
}%
\resizebox{0.5\linewidth}{!}{
%
%
\begin{psfrags}%
\psfragscanon%
%
\psfrag{s08}[b][b]{\setlength{\tabcolsep}{0pt}\begin{tabular}{c}d\end{tabular}}%
\psfrag{s10}[][]{\setlength{\tabcolsep}{0pt}\begin{tabular}{c} \end{tabular}}%
\psfrag{s11}[][]{\setlength{\tabcolsep}{0pt}\begin{tabular}{c} \end{tabular}}%
\psfrag{s12}[l][l]{L}%
\psfrag{s13}[l][l]{H}%
\psfrag{s14}[l][l]{J}%
\psfrag{s15}[l][l]{L}%
%
\psfrag{x01}[t][t]{-2}%
\psfrag{x02}[t][t]{0}%
\psfrag{x03}[t][t]{2}%
\psfrag{x04}[t][t]{4}%
%
\psfrag{v01}[r][r]{0}%
\psfrag{v02}[r][r]{0.5}%
\psfrag{v03}[r][r]{1}%
\psfrag{v04}[r][r]{1.5}%
\psfrag{v05}[r][r]{2}%
%
\resizebox{6cm}{!}{\includegraphics{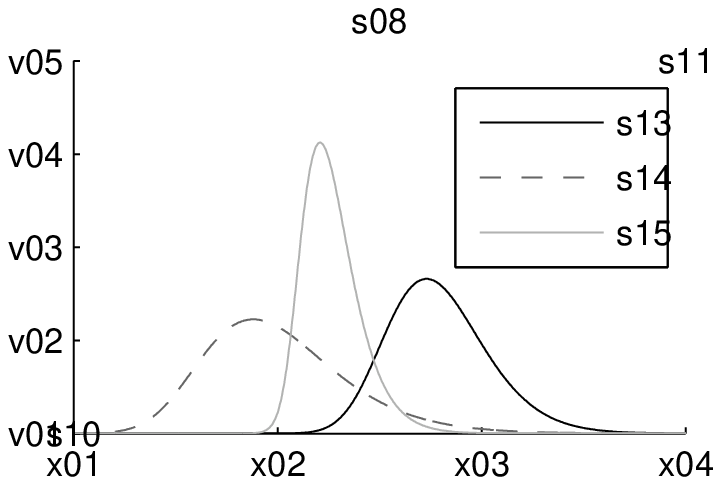}}%
\end{psfrags}%
%
}%
\end{center}
\vspace{-0.5cm}
\caption{Marginal distributions for the twelve test cases in the simulation study. In Panel a and Panel b, the approximate covariance range is $3.5$ and in Panel c and d, the range is $35$.}
\label{paperD:fig:params}
\end{figure}

\begin{table}[t]\centering\scriptsize
\begin{tabular}{|@{\hspace{2pt}}c@{\hspace{2pt}}
|@{\hspace{1pt}}c@{}r@{.}l@{ }r@{.}l@{ }r@{.}l@{}
|@{\hspace{1pt}}c@{}r@{.}l@{ }r@{.}l@{ }r@{.}l@{}
|@{\hspace{1pt}}c@{}r@{.}l@{ }r@{.}l@{ }r@{.}l@{}
|@{\hspace{1pt}}c@{}r@{.}l@{ }r@{.}l@{ }r@{.}l@{}
|@{\hspace{1pt}}r@{}r@{.}l@{ }r@{.}l@{ }r@{.}l@{}|}
\hline
& \multicolumn{7}{c}{$\kappa$} 
& \multicolumn{7}{c}{$\tau$} 
& \multicolumn{7}{c}{$\sigma$} 
& \multicolumn{7}{c}{$\mu$} 
& \multicolumn{7}{c|}{$\gamma$}\\
\hline
A 
&  1 &  (0&95 & 1&00 & 1&06)
&  2 &  (1&63 & 2&03 & 3&02)
&  1 &  (0&78 & 0&98 & 1&13)  
&  0 &  (-0&07 & -0&01 & 0&06)
&  0 &  (-0&06 & 0&00 & 0&07) \\[1ex]

B 
&  1 &  (0&96 & 1&00 & 1&05)  
&  2 &  (1&68 & 1&99 & 2&41)
& $\frac{1}{2}$ & (0&42 & 0&50 & 0&57)
& $\frac{1}{2}$ & (0&43 & 0&50 & 0&57)
&  0 & (-0&06 & 0&00 & 0&07)\\[1ex]

C 
&  1 &  (0&96 & 1&00 & 1&05)  
&  1 &  (0&85 & 0&99 & 1&21)
&  1 &  (0&87 & 1&00 & 1&11)  
&  0 &  (-0&07 & 0&00 & 0&05)
&  0 &  (-0&04 & 0&00 & 0&05) \\[1ex]

D 
&   1 &  (0&96 & 1&00 & 1&04)  
&  1 &  (0&90 & 1&00 & 1&14)  
&  1 &  (0&90 & 1&00 & 1&10)  
&  1 &  (0&87 & 1&00 & 1&13)
&  -1 &  (-1&11 & -1&00 & -0&89)\\[1ex]

E 
&   1 &  (0&97 & 1&00 & 1&02)  
& $\frac{1}{2}$ &  (0&45 & 0&49 & 0&54)  
&  1 &  (0&93 & 1&00 & 1&08)  
&  0 &  (-0&06 & 0&00 & 0&06)
&  0 &  (-0&01 & 0&00 & 0&01)\\[1ex]

F 
&   1 &  (0&98 & 1&00 & 1&01)  
& $\frac{1}{2}$ &  (0&46 & 0&50 & 0&54)  
&  1 &  (0&91 & 1&00 & 1&08)  
&  1 &  (0&89 & 1&00 & 1&11)
&  -1 &  (-1&09 & -1&01 & -0&92)\\[1ex]

G 
& $\frac{1}{10}$ &  (0&09 & 0&10 & 0&11)
&  1 &  (0&86 & 1&00 & 1&24)  
&  1 &  (0&87 & 1&00 & 1&11)  
&  0 &  (-0&06 & 0&00 & 0&06)
&  0 &  (-0&04 & 0&00 & 0&04)\\[1ex]

H 
&  $\frac{1}{10}$ &  (0&09 & 0&10 & 0&11)
&  1 &  (0&89 & 0&99 & 1&13)  
& $\frac{1}{2}$ &  (0&45 & 0&50 & 0&54)
& $\frac{1}{2}$ &  (0&44 & 0&50 & 0&56)
&  0 &  (-0&05 & 0&01 & 0&13)\\[1ex]

I 
&  $\frac{1}{10}$ &  (0&10 & 0&10 & 0&10)
& $\frac{1}{2}$ &  (0&45 & 0&49 & 0&54)  
&  1 &  (0&91 & 1&01 & 1&09)  
&  0 &  (-0&06 & 0&00 & 0&07)
&  0 &  (-0&01 & 0&00 & 0&01)\\[1ex]

J 
&  $\frac{1}{10}$ &  (0&10 & 0&10 & 0&10)
& $\frac{1}{2}$ &  (0&46 & 0&50 & 0&54)  
&  1 &  (0&91 & 0&99 & 1&08)  
&  1 &  (0&90 & 1&01 & 1&13)
&  -1 &  (-1&09 & -1&01 & -0&93)\\[1ex]

K 
&  $\frac{1}{10}$ &  (0&10 & 0&10 & 0&10)
& $\frac{1}{3}$ &   (0&31 & 0&33 & 0&36) 
&  1 &  (0&91 & 0&99 & 1&06) 
&  0 &  (-0&07 & 0&00 & 0&07)
&  0 &  (-0&00 & 0&00 & 0&00)\\[1ex]

L 
&  $\frac{1}{10}$ &  (0&09 & 0&10 & 0&12)
& $\frac{1}{3}$ &  (0&33 & 0&36 & 0&43)  
& $\frac{1}{2}$ &  (0&42  & 0&47 & 0&51) 
& $\frac{1}{2}$ &  (0&39 & 0&50 & 0&52)
&  0 &  (-0&05 & 0&00 & 0&13)\\[1ex]
\hline
\end{tabular}
\caption{Parameter settings for the twelve cases shown in Figure \ref{paperD:fig:params} the estimation procedure is tested for. In the parentheses, the $10\%$ $50\%$, and $90\%$ percentiles of 500 Monte Carlo samples are shown. Note that most estimates seem to be unbiased, perhaps with the exception of the estimates of $\tau$ in case L.}
\label{paperD:tab1}
\end{table}

In this section, a simulation study is performed to test the accuracy of the parameter estimation algorithm presented above. The algorithm is tested for twelve different parameter settings corresponding to marginal distributions shown in Figure \ref{paperD:fig:params} for processes in one dimension with $\alpha=2$. For Mat\'{e}rn covariance functions, one sometimes defines the approximate range as $r = \sqrt{8\nu}\kappa^{-1}$, which is the value where the correlation is approximately $0.1$. For the first six test cases, we have $\kappa=1$ which corresponds to an approximate range of $3.5$, and for the last six cases we have $\kappa=0.1$ which corresponds to an approximate range of $35$. For each value of $\kappa$, three symmetric distributions and three asymmetric distributions are used. In Figure \ref{paperD:fig:params}, the distributions for the short range are shown in the two upper panels, and the distributions for the long range are shown in the two bottom panels. 

For each set of parameters, $500$ data sets are simulated using Algorithm \ref{paperD:alg:sampel}, where each data set contains $1000$ equally spaced observations on $[1, 1000]$. The basis used in the Hilbert space approximations consists of $1000$ piecewise linear basis functions centered at $1, 2, \ldots, 1000$. For each data set, the starting value for $\kappa$ is set to $\sqrt{8\nu}\hat{r}^{-1}$, where $\hat{r}$ is the approximate range for the empirical covariance function for the data set. To obtain good starting values for the other parameters, an initial run of the EM estimator is made with $\kappa$ fixed to the starting value and where the starting values for $\mu$ and $\gamma$ are drawn independently from a $\pN(0,1)$ distribution, and the starting values for $\sigma$ and $\tau^{-1}$ are drawn from a $\chi^2(1)$-distribution. After $100$ steps, this initial run is ended, and the estimates are used as starting values for the full EM-estimator.

In Table~\ref{paperD:tab1}, the $10\%$ $50\%$, and $90\%$ percentiles of 500 Monte Carlo samples are shown for each parameter setting, together with the true values of the parameters. One can note that all estimates are more or less unbiased and have fairly small variances, indicating that the estimation procedure works as intended. The only case where the estimator seems to have a bias is in case L, where most of the estimated values of $\tau$ are above the true value. The cause of this bias is probably that the estimation procedure is not very stable for small values of $\tau$ because some of the expectations $\pE(\mv{\Gamma}_i^{-1}|\star)$ can be infinite in this case. More precisely, for $\tau<3/(2\min(\mv{a}_1,\ldots,\mv{a}_n))$, the likelihood is unbounded for any $\bar{\gamma} = \mv{\Lambda}_i/\mv{a}_i$ and the ML procedure thus has to be modified. To improve the stability of the algorithm, the expectations $\pE(\mv{\Gamma}_i^{-1}|\star)$ are truncated to $1000$ in the first iteration, and for each iteration this bound is made larger so that it has little to no effect after a few hundred iterations of the algorithm. This greatly improves the stability for $\tau<3/(2\min(\mv{a}_1,\ldots,\mv{a}_n))$, but it is left for future research to justify this modified maximum likelihood procedure theoretically, to derive large sample properties of the estimator, and to investigate other improvements for the case of small values of $\tau$.

It should finally be noted that the parameters are estimated assuming the same finite element approximation as is used for simulating the data. Estimating the model parameters using a different numbers of basis functions in the approximation can possibly give biased estimates, as the parameters are estimated to maximize the likelihood for the approximate model instead of the exact SPDE. The size of this bias depends on the specific parameters of the model, and especially on the true covariance range in relation to the spacing of the basis functions, as discussed in \cite{bolin09b} in the case of Gaussian models. It is, however, outside the scope of this work to investigate this issue further here. 

\section{Discussion and extensions}\label{paperD:sec:discussion}
We have showed how the SPDE approach by \cite{lindgren10} can be extended to the case of Laplace noise and how this can be used to obtain an efficient estimation procedure as well as an accurate estimation technique for the Laplace moving average models. This is indented as a demonstration that the methods in \cite{lindgren10} are applicable to more general situations than the ordinary Gaussian models. There are also a number of extensions that can be made to this work which are discussed below. 

First of all, the Hilbert space approximation technique in Section \ref{paperD:sec:hilbert} was derived using theory for L\'{e}vy processes of type G, and although we only used this for the case of Laplace noise, the methods work equally well for this larger class of models. All that is changed are the distributions of the integrals conditionally on the variance process. These techniques are also applicable to the case when more general SPDEs are used, one could for example use the nested SPDEs by \cite{bolin09b} to achieve more general covariance structures without any additional work needed, or one could include drift terms in the operator on the left-hand side to mimic the effects of asymmetric kernels in the Laplace moving average models. The methods are in fact not restricted to $\R^d$ or stationary SPDEs, but can be extended to non-stationary SPDEs on general Riemann manifolds.

Secondly, the estimation procedure in Section \ref{paperD:sec:estimation} assumed that one basis function was used for each observation of the process. The reason being that this gives us a one-to-one correspondence between the observations and the Laplace variables $\mv{\Lambda}$ which simplified the estimation procedure. For practical applications this is not ideal as one would like to be able to choose the basis independently of the measurement locations, and it would also be useful if one could assume that the measurements are taken under measurement noise. If the estimation procedure could be extended to handle these cases, the practical usefulness of these models would greatly improve. 

As mentioned in Section \ref{paperD:sec:study}, the estimation procedure is sensitive to the value of $\tau$. Too large values will result in a model which is very similar to a standard Gaussian model, and it might be difficult to accurately estimate the parameters in this case without a very large data set. This is not a big problem as if the data is Gaussian, one should not use these models but a standard Gaussian model. The estimation procedure is also unstable for small values of $\tau$, and modifications to further improve the stability in this case are currently being investigated.
   
\section*{Acknowledgements}   
The author is grateful to Krzysztof Podg\'orski for many helpful comments and discussions regarding the theoretical aspects of this work, to Daniel Simpson for providing some of his code for the Krylov subspace methods used in Section \ref{paperD:sec:sampel}, and to Jonas Wallin for numerous discussions regarding the parameter estimation problem and for suggesting the truncation of the expectations mentioned at the end of Section \ref{paperD:sec:study}.

\bibliographystyle{imsart-nameyear}
\bibliography{Journals_abrv,completebib}

\end{document}